\documentclass[review,sort&compress,1p]{elsarticle}

\usepackage{amssymb}
\usepackage{amsthm}
\usepackage[fleqn]{amsmath}
\usepackage{amssymb,graphicx}
\usepackage{amsfonts}
\usepackage{mathrsfs}
\usepackage{epsfig}
\usepackage{cite}
\usepackage[active]{srcltx}
\usepackage{cmap}
\usepackage{amsthm}
\usepackage[active]{srcltx}
\usepackage{amsmath}
\usepackage{epstopdf}
\usepackage{curves}
\usepackage{amssymb}
\usepackage{mathrsfs}
\usepackage{epsfig}
\usepackage{color}
\usepackage{cite}
\usepackage{graphicx}
\usepackage{subfigure}
\usepackage{amsmath}
\usepackage{amsthm}
\usepackage[dvipsnames]{xcolor}
\usepackage{graphicx}
\usepackage{fancybox}
\usepackage{array}
\usepackage{hyperref}
\usepackage[fleqn]{amsmath}
\usepackage{amssymb,amsthm,bm}
\usepackage{latexsym}
\usepackage{graphicx}
\usepackage{psfrag}
\usepackage{color}
\usepackage{subfigure}
\usepackage{hyperref}
\usepackage{amssymb}
\usepackage{amsthm}
\usepackage[fleqn]{amsmath}
\usepackage{amssymb,graphicx}
\usepackage{amsfonts}
\usepackage{mathrsfs}
\usepackage{epsfig}
\usepackage{cite}
\usepackage[active]{srcltx}
\usepackage{cmap}
\usepackage{amsthm}
\usepackage[active]{srcltx}
\usepackage{amsmath}
\usepackage{epstopdf}
\usepackage{curves}
\usepackage{amssymb}
\usepackage{mathrsfs}
\usepackage{epsfig}
\usepackage{color}
\usepackage{cite}
\usepackage{graphicx}
\usepackage{subfigure}
\usepackage{amsmath}
\usepackage{amsthm}
\usepackage[dvipsnames]{xcolor}
\usepackage{graphicx}
\usepackage{fancybox}
\usepackage{array}
\usepackage{hyperref}
\usepackage[fleqn]{amsmath}
\usepackage{amssymb,amsthm,bm}
\usepackage{latexsym}
\usepackage{graphicx}
\usepackage{psfrag}
\usepackage{color}
\usepackage{subfigure}
\usepackage{hyperref}

\newdefinition{rmk}{Remark}
\newproof{pf}{Proof}
\newproof{pot}{Proof of Theorem \ref{thm2}}
\newtheorem{theorem}{Theorem}[section]
\newtheorem{lemma}{Lemma}[section]

\newtheorem{remark}{Remark}[section]

\begin{document}
\begin{frontmatter}
\title{Multi-scale Modeling for  Piezoelectric Composite Materials}

\author[lbl1]{Qian Zhang}
\ead{qzhangxd@ust.hk}
\author[lbl2]{Xingye Yue\corref{cor1}}
\ead{xyyue@suda.edu.cn}

\

\cortext[cor1]{Corresponding author}

\address[lbl1]{Department of Mathematics, The Hong Kong University of
                       Science and Technology Clear Water Bay, Kowloon, Hong Kong, China}

\address[lbl2]{Department of Mathematics, Soochow University, Suzhou, China}
\begin{abstract}
In this paper, we focus on multi-scale modeling and simulation of piezoelectric composite materials.
A multi-scale model for piezoelectric composite materials under the framework of Heterogeneous Multi-scale Method(HMM) is proposed. For materials with periodic microstructure, macroscopic model is derived from microscopic
model of piezoelectric composite material by asymptotic expansion.
Convergence analysis under the framework of homogenization theory is carried out. Moreover, error
estimate between HMM solutions and homogenization solutions is derived. A 3-D numerical example of 1-3 type piezoelectric composite materials is employed to verify the error estimate.
\end{abstract}

\end{frontmatter}

\section{Introduction}

Piezoelectricity is the ability of some materials to generate an electric field or electric potential in response to mechanical strain applied. The piezoelectric effect is reversible in that materials exhibiting the direct piezoelectric effect (the production of an electric potential when stress is applied) also exhibit the reverse piezoelectric effect (the production of stress and/or strain when an electric field is applied). The effect is found useful in applications such as the production and detection of sound, generation of high voltages, electronic frequency generation, microbalances, and ultra fine focusing of optical assemblies. It is also the basis of a number of scientific instrumental techniques with atomic resolution, the scanning probe microscopies such as STM, AFM, MTA, SNOM, etc., and everyday uses such as acting as the ignition source for cigarette lighters and push-start propane barbecues. Moreover, the piezoelectric effect and the reverse piezoelectric effect can be reflected  by the following two  constitutive relationships respectively
\begin{align}
&D_{m} = \epsilon_{mn}\,E_{n} + d_{mkl}\,\sigma_{kl}
&S_{ij} = s_{ijkl}\,\sigma_{kl} + d_{mij}\,E_{m}
\end{align}
where $\mathbf{D}$ is electric displacement, $\epsilon_{mn}$ is permittivity, $\mathbf{E}$ is electric field strength, $S_{ij}$ is strain, $s_{ijkl}$ is compliance and $\sigma_{kl}$ is stress, $d_{mkl}$ is piezoelectric constant.

It is apparent that some composite materials could be designed so as to attain
 properties desired in industry. In this way, the best properties from each constituent phase within the
composite may be utilized to create an improved material. In general, by replacing a
portion of a piezoceramic with a lightweight, flexible, polymer, the resulting
density, acoustic impedance, mechanical quality factor and dielectric constant can be
decreased. If the phases can be arranged so the piezoelectric charge coefficients of the composite are maintained at reasonable levels, its voltage coefficients
can be substantially improved. As a result, the merit of
the piezocomposites can actually surpass those of single phase materials.

Therefore, according to aspects mentioned above, multi-scale modeling is necessary to describe the behavior of piezoelectric composites for the following two reasons. On one hand, in engineering, we are interested in relationship between micro-structure of
the composite materials and macro-properties of the composite materials because in this way we can make the composite materials with desired macro-properties by controlling its micro-structure. On the other hand, in numerical computation, the computation of the parameters describing a micro-nonhomogeneous medium is an extremely difficult task, since the coefficients of the corresponding differential equations are given by rapidly oscillating functions. Many methods has been raised to get the multi-scale model, such as MsFEM, RVE, HMM and etc.

A number of work have been done in this field by former researchers. Most of them focus on the piezoelectric composite materials with some special microstructures. The modeling of the piezoelectric composites with shell, perforated, periodic micro-structure with theoretical analysis has been held by \citet{Gerges} and \citet{bone}.
The modeling of $1-3$ type piezoelectric composite without theoretical analysis has been held by \citet{Harald} using RVE.

In this paper, a multi-scale model under the framework of HMM is designed. Corresponding theoretical analysis is derived. Differing from the elliptic equations whose solution is a minimum point of its energy functional, the solution of piezoelectric equations is a saddle point
of its energy functional, which is caused by the coupling of mechanical field and electrical field.
This difference brought main obstacle for theoretical analysis. We also deduced the macroscopic model from the microscopic model of piezoelectric composite material for materials with periodic microstructure and carried out the corresponding convergence analysis under the framework of homogenization. For the analysis of convergence, noticing that the regularity of the solutions of cell
problems in the model cannot reach $W^{1,\infty}$, which makes the traditional treatment to this kind of
problem not work on our problem. To solve this problem, we employ two useful lemmas, whose original idea is from \citet{Suslina} under the assumption that the regularity of solutions of cell problems in the model can reach $L^{\infty}$. Moreover, error
estimate between HMM solutions and homogenization solutions is derived. A 3-D numerical example was given out to verify the error estimate.

The rest of the paper is organized as follows. In Section 2, we introduction the multi-scale problem of piezoelectric composite material briefly. In Section 3, a multi-scale model under the framework of HMM is designed. The macroscopic model from the microscopic model of piezoelectric composite material for materials with periodic microstructure is derived by asymptotic expansion in Section 4. We also give out the corresponding convergence analysis for the model we derived under the framework of homogenization in Section 4. In Section 5, both error
estimate between HMM solutions and homogenization solutions and error estimate of the effective coefficients is derived. Numerical examples are employed to verify the error estimate in Section 6. The paper concludes in Section 7.

\section{Piezoelectric Equations}
Under the action of applied volume loading $\mathbf{f}\in L^2(\Omega)$ and without electric charges, the electroelastic state of piezoelectric medium is governed by the following system of equations.
\begin{enumerate}
\item[$\bullet$]Equations of motion:
\begin{align}
-\frac{\partial{\sigma_{ij}}}{\partial{x_{j}}}=f_{i}\qquad \mbox{in\quad $\Omega$}
\end{align}
\item[$\bullet$]Maxwell's equations (in the quasistatic approximation):
\begin{align}
        \frac{\partial{D_{i}}}{\partial{x_{i}}}={0}\qquad E_{i}=-\frac{\partial\varphi}{\partial{x_{i}}}\qquad \mbox{in\quad $ \Omega$}
\end{align}
\item[$\bullet$]Constitutive relations:
\begin{align}
        \sigma_{ij} & =  c_{ijkn}\frac{\partial{u_{k}}}{\partial{x_{n}}}-e_{kij}E_{k}\\
        D_{i}       & =  e_{ikn}\frac{\partial{u_{k}}}{\partial{x_{n}}}+\epsilon_{ij}E_{j}
\end{align}
\end{enumerate}
Without lose of generalization, we just take the Dirichlet Boundary condition as follows.
\begin{align}
  u_{i}=0,\quad\varphi=0\qquad\mbox{on\quad $\partial\Omega$}
  \end{align}
where $\sigma_{ij}$ is stress tensor, $\mathbf{u}$ is elastic displacement field, $\mathbf{D}$ is electrical displacement, $\mathbf{E}$ is electric field, $\varphi$ is potential field.

Hence, characteristics of the piezoelectric material are given by elastic tensor $(c_{ijkn})$, dielectric tensor $(\epsilon_{ij})$ and piezoelectric tensor $(e_{kij})$. These three tensors have the following properties:
\begin{enumerate}
\item[$\bullet$]The elastic tensor $(c_{ijkn})$ is symmetric and positive defined, that is,
\begin{equation*}
c_{ijkl}=c_{jikl}=c_{klij}
\end{equation*}
and there exists $\alpha >0$, such that $c_{ijkl}X_{ij}X_{kl}\geq\alpha X_{ij}X_{kl}$, $\forall X_{ij}=X_{ji}\in R$
\item[$\bullet$]The dielectric tensor $(\epsilon_{ij})$ is symmetric and positive defined, that is,
\begin{equation*}
\epsilon_{ij}=\epsilon_{ij}
\end{equation*}
and there exists $\beta>0$, such that $\epsilon_{ij}X_{i}X_{j}\geq\beta X_{i}X_{j}$, $\forall X_{i}\in R$
\item[$\bullet$]The piezoelectric tensor $(e_{kij})$ is symmetric in the sense that $e_{kij}=e_{kji}$
\item[$\bullet$]$c_{ijkl} \,\in\, L^{\infty}(\Omega)$, $d_{ij}\,\in\, L^{\infty}(\Omega)$, $e_{ijk}\,\in\, L^{\infty}(\Omega)$
\end{enumerate}
Thus, in the rest parts of this paper, we consider the following equation system,
\begin{align}\label{micro equ}
            \left\{\begin{array} {l@{\quad}l} \displaystyle
                -\frac{\partial{\sigma_{ij}^{(\varepsilon)}}}{\partial{x_{j}}}=f_{i}\qquad
                &\mbox{ ${x}\in\Omega\subset R^{d}$}\\
                \displaystyle    \frac{\partial{D_{i}^{(\varepsilon)}}}{\partial{x_{i}}}={0}\qquad&\mbox{ ${x}\in \Omega\subset R^{d}$}\\
                \displaystyle    \sigma_{ij}^{(\varepsilon)}=c_{ijkn}^{(\varepsilon)}\frac{\partial{u_{k}^{(\varepsilon)}}}{\partial{x_{n}}}+e_{kij}^{(\varepsilon)}\frac{\partial{\varphi^{(\varepsilon)}}}{\partial{x_{k}}}\\
                \displaystyle
                D_{i}^{(\varepsilon)}=e_{ikn}^{(\varepsilon)}\frac{\partial{u_{k}^{(\varepsilon)}}}{\partial{x_{n}}}-\epsilon_{ij}^{(\varepsilon)}\frac{\partial{\varphi^{(\varepsilon)}}}{\partial{x_{j}}}\\
            \displaystyle
            u_{i}^{(\varepsilon)}=0��\qquad
            \varphi^{(\varepsilon)}=0\qquad&\mbox{  $x\in\partial\Omega$}
            \end{array}\right.
\end{align}
where $\varepsilon\ll 1$ is the characteristic non-homogeneity dimension.

Moreover, through this paper, Latin indices and exponents take their values in the set \{1,2,3\}, if there is no special illustration. The average symbol $\langle\quad\cdot\quad\rangle_\Box$ is defined as
\begin{align*}
\langle\quad\cdot\quad\rangle_\Box\ = \frac{1}{|\Box|}\int_{\square}\cdot \quad dy,
\end{align*}
where $|\Box|$ is the volume of $\Box$. Let $\Omega$ be a polyhedral domain in $R^{d}$ with a
Lipschitz boundary $\partial \Omega$ whose unit outer normal is
denoted by $\bf n$. In the derivations below, we use
$L^{2}(\Omega)$ based Sobolev spaces $H^{k}(\Omega)$ equipped with
norms and seminorms
$$\parallel u\parallel_{k,\Omega}=\left(\int_{\Omega}\sum_{|\alpha|\leq k}|D^{\alpha}u|^{2}\right)^{1/2},
\qquad \mid u\mid_{k,\Omega}=\left(\int_{\Omega}\sum_{|\alpha|=
k}|D^{\alpha}u|^{2}\right)^{1/2} $$.

\section{HMM modeling of piezoelectric composite materials}

HMM (Heterogeneous multiscale method) by \citet{E} is a general methodology
for designing sublinear algorithms by exploiting scale separation and other
special features of the problem. It consists of two components: selection of a macroscopic
solver and estimating the missing macroscale data by solving locally the fine
scale problem.

For (\ref{micro equ}), the macroscopic solver can be chosen as standard piecewise linear finite element method on a triangulation $\mathscr{T}_{H}$ of element size H which should resolve the macro-scale features of $c_{ijkn}^{(\varepsilon)},e_{kij}^{(\varepsilon)}$ and $\epsilon_{ij}^{(\varepsilon)}$. The missing data are the effective stiffness matrices at this scale.
Assuming that the effective coefficients at this scale are $c^{H}_{ijkn}$, $e^{H}_{ijk}$ and $\epsilon^{H}_{ij}$, if we know $c^{H}_{ijkn}$, $e^{H}_{ijk}$ and $\epsilon^{H}_{ij}$ explicitly, macroscopic piezoelectric equation system could be solved by FEM in the following variational form.

Find $\mathbf{u}\in X_{H}\times X_{H}\times X_{H}$,$\varphi\in X_{H}$ s.t.
    \begin{align}\label{hmm macroweak}
      \left\{\begin{array} {l@{\quad}l} \displaystyle
          c_{H}(\mathbf{u},\mathbf{v})+e_{H}(\mathbf{v},\varphi)=(\mathbf{f},\mathbf{v}),\qquad&\forall\quad\mathbf{v}\in
          X_{H}\times X_{H}\times X_{H}\\
          \displaystyle
          -e_{H}(\mathbf{u},\psi)+d_{H}(\varphi,\psi)=0,\qquad
          &\forall\quad\psi\in X_{H}
        \end{array}\right.
    \end{align}
where
 \begin{align}\label{int}
      c_{H}(\mathbf{u},\mathbf{v})&=\int_{\Omega}c_{ijkl}^{H}s_{ij}(\mathbf{u})s_{kl}(\mathbf{v})dx \simeq \sum_{K\in \mathscr{T}_{H}}\mid K \mid \sum_{x_{\alpha}\in K}w_{\alpha}(c_{ijkl}^{H}s_{ij}(\mathbf{u})s_{kl}(\mathbf{v}))(x_\alpha)\nonumber\\
      e_{H}(\mathbf{v},\varphi)&=\int_{\Omega}e_{ijk}^{H}s_{ij}(\mathbf{v})\partial_{k}\varphi dx \simeq \sum_{K\in \mathscr{T}_{H}}\mid K \mid \sum_{x_{\alpha}\in K}w_{\alpha}(e_{ijk}^{H}s_{ij}(\mathbf{v})\partial_{k}\varphi)(x_\alpha)\\
      d_{H}(\varphi,\psi)&=\int_{\Omega}\epsilon_{ij}^{H}\partial_{i}\varphi\partial_{j}\psi dx \simeq \sum_{K\in \mathscr{T}_{H}}\mid K \mid \sum_{x_{\alpha}\in K}w_{\alpha}(\epsilon_{ij}^{H}\partial_{i}\varphi\partial_{j}\psi)(x_\alpha)\nonumber
    \end{align}
where ${x_\alpha}$ and ${w_\alpha}$ are the quadrature points and weights in K, $\mid K\mid$ is the volume of K.

However, we can not get $c^{H}_{ijkn}$, $e^{H}_{ijk}$ and $\epsilon^{H}_{ij}$ explicitly in most cases. In the absence of explicit knowledge of $c^{H}_{ijkn}$, $e^{H}_{ijk}$ and $\epsilon^{H}_{ij}$, we approximation $c_{H}(\mathbf{u},\mathbf{v})$, $e_{H}(\mathbf{v},\varphi)$ and $d_{H}(\varphi,\psi)$ by solving two microscopic problems as follows on the samples we have chosen.
\begin{align}\label{hmm 1}
          \left\{\begin{array} {l@{\quad}l}
          \displaystyle-\frac{\partial{\sigma_{ij}^{(\varepsilon)}}}{\partial{x_{j}}}=0\qquad&\mbox{ $x\in I_{\delta}(x_{\alpha})$}\\
           \displaystyle     \frac{\partial{D_{i}^{(\varepsilon)}}}{\partial{x_{i}}}=0 \qquad&\mbox{ $x\in I_{\delta}(x_{\alpha})$}\\
            \displaystyle    \sigma_{ij}^{(\varepsilon)}=c_{ijkn}^{(\varepsilon)}\frac{\partial{(v_{\alpha})_{k}^{(\varepsilon)}}}{\partial{x_{n}}}+e_{kij}^{(\varepsilon)}\frac{\partial{\varphi_{\alpha v}^{(\varepsilon)}}}{\partial{x_{k}}}\\
            \displaystyle D_{i}^{(\varepsilon)}=e_{ikn}^{(\varepsilon)}\frac{\partial{(v_{\alpha})_{k}^{(\varepsilon)}}}{\partial{x_{n}}}-\epsilon_{ij}^{(\varepsilon)}\frac{\partial{\varphi_{\alpha v}^{(\varepsilon)}}}{\partial{x_{j}}}\\
              \displaystyle \mathbf{v}_{\alpha}^{(\varepsilon)}=\mathbf{V}_{\alpha}(x)\qquad
              \varphi_{\alpha v}^{(\varepsilon)}=0\qquad&\mbox{ $x\in\partial I_{\delta}(x_{\alpha})$}\\
            \end{array}\right.
\end{align}
and
\begin{align}\label{hmm2}
            \left\{\begin{array} {l@{\quad}l} \displaystyle -\frac{\partial{\sigma_{ij}^{(\varepsilon)}}}{\partial{x_{j}}}=0\qquad&\mbox{ $x\in I_{\delta}(x_{\alpha})$}\\
                \displaystyle    \frac{\partial{D_{i}^{(\varepsilon)}}}{\partial{x_{i}}}=0\qquad&\mbox{ $x\in I_{\delta}(x_{\alpha})$}\\
                \displaystyle    \sigma_{ij}^{(\varepsilon)}=c_{ijkn}^{(\varepsilon)}\frac{\partial{({v}_{\alpha\varphi})_{k}^{(\varepsilon)}}}{\partial{x_{n}}}+e_{kij}^{(\varepsilon)}\frac{\partial{\varphi_{\alpha}^{(\varepsilon)}}}{\partial{x_{k}}}\\
                \displaystyle D_{i}^{(\varepsilon)}=e_{ikn}^{(\varepsilon)}\frac{\partial{({v}_{\alpha\varphi})_{k}^{(\varepsilon)}}}{\partial{x_{n}}}-\epsilon_{ij}^{(\varepsilon)}\frac{\partial{\varphi_{\alpha}^{(\varepsilon)}}}{\partial{x_{j}}}\\
                \displaystyle \mathbf{v}_{\alpha\varphi}^{(\varepsilon)}=0\qquad
                \varphi_{\alpha}^{(\varepsilon)}=\varphi_{\alpha}(x)
                \qquad&\mbox{ $x\in\partial I_{\delta}(x_{\alpha})$}\\
              \end{array}\right.
          \end{align}
where $I_{\delta}({x_{\alpha}})={x_{\alpha}}+[-\frac{\delta}{2},\frac{\delta}{2}]^{d}$ is a cubic of size $\delta$ centered at ${x_{\alpha}}$ and $\mathbf{V}_{\alpha}$ is the linear approximation of $\mathbf{v}$ at $I_{\delta}(x_{\alpha})$. For the macroscopic finite element space $X_{H}$ we have chosen, $\mathbf{V}_{\alpha}$ is $\mathbf{v}$. Similarly, $\varphi_{\alpha}(x)$ is the linear approximation of $\varphi$ at $I_{\delta}(x_{\alpha})$.\\

\begin{tabular}[t]{c}
  \begin{picture}(150,112)(-30,20)
\put(50,50){\line(2,3){50}} \put(50,50){\line(1,0){100}}
\put(150,50){\line(-2,3){50}}
\put(95,75){\small{K}}
\multiput(66,56)(10,0){2}{\line(0,1){10}}
\multiput(66,56)(0,10){2}{\line(1,0){10}} \put(71,61){\circle*{2}}
\put(78,59){\scriptsize{$I_{\delta}(x_{\alpha})$}}
\multiput(123,56)(10,0){2}{\line(0,1){10}}
\multiput(123,56)(0,10){2}{\line(1,0){10}} \put(128,61){\circle*{2}}
\multiput(94,95)(10,0){2}{\line(0,1){10}}
\multiput(94,95)(0,10){2}{\line(1,0){10}}
 \put(99,100){\circle*{2}}
\put(280,105){\vector(-1,0){150}}
\put(145,90){element}
\put(150,110){unit}
\put(-50,20){Figure 1. Illustration of HMM for solving (\ref{micro equ}). The dots are the quadrature points in (\ref{int}).}
\end{picture}
\begin{picture}(120,112)(-80,-20)
{\includegraphics[height=35mm,width =45mm]{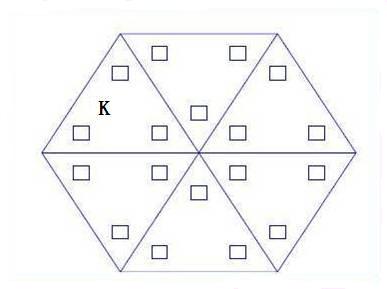}}
\end{picture}
\end{tabular}

we define the corresponding bilinear forms as follows.

For any $\mathbf{u}\in X_{H}\times X_{H}\times X_{H}$,$\mathbf{v}\in X_{H}\times X_{H}\times X_{H}$, $\varphi\in X_{H}$,$\psi\in X_{H}$
\begin{align*}
    \displaystyle c_{H}(\mathbf{u},\mathbf{v})&=\sum_{K\in \mathscr{T}_{H}}\mid K \mid \sum_{x_{\alpha}\in K}w_{\alpha}\frac{1}{\mid I_{\delta}(x_{\alpha})\mid }\int_{I_{\delta}(x_{\alpha})}(c_{ijkl}^{(\varepsilon)}s_{ij}(\mathbf{u}_{\alpha}^{(\varepsilon)})s_{kl}(\mathbf{v}_{\alpha}^{(\varepsilon)})+e_{ijk}^{(\varepsilon)}s_{ij}(\mathbf{u}_{\alpha}^{(\varepsilon)})\partial_{k}\varphi_{\alpha v}^{(\varepsilon)})dx\\
    \displaystyle e_{H}(\mathbf{v},\varphi)&=\sum_{K\in \mathscr{T}_{H}}\mid K \mid\sum_{x_{\alpha}\in K}w_{\alpha}\frac{1}{\mid I_{\delta}(x_{\alpha})\mid }\int_{I_{\delta}(x_{\alpha})}(e_{ijk}^{(\varepsilon)}s_{ij}(\mathbf{v}_{\alpha}^{(\varepsilon)})\partial_{k}\varphi_\alpha^{(\varepsilon)}-\epsilon_{ij}^{(\varepsilon)}\partial_{i}\varphi_{\alpha v}^{(\varepsilon)}\partial_{j}\varphi_\alpha^{(\varepsilon)})dx\\
     \displaystyle e_{H}(\mathbf{v},\varphi)&=\sum_{K\in \mathscr{T}_{H}}\mid K \mid \sum_{x_{\alpha}\in K}w_{\alpha}\frac{1}{\mid I_{\delta}(x_{\alpha})\mid }\int_{I_{\delta}(x_{\alpha})}(c_{ijkl}^{(\varepsilon)}s_{ij}(\mathbf{v}_{\alpha}^{(\varepsilon)})s_{kl}(\mathbf{v}_{\alpha\varphi}^{(\varepsilon)})+e_{ijk}^{(\varepsilon)}s_{ij}(\mathbf{v}_{\alpha}^{(\varepsilon)})\partial_{k}\varphi_\alpha^{(\varepsilon)})dx\\
     \displaystyle d_{H}(\varphi,\psi)&=-\sum_{K\in \mathscr{T}_{H}}\mid K \mid\sum_{x_{\alpha}\in K}w_{\alpha}\frac{1}{\mid I_{\delta}(x_{\alpha})\mid }\int_{I_{\delta}(x_{\alpha})}(e_{ijk}^{(\varepsilon)}s_{ij}(\mathbf{v}_{\alpha\varphi}^{(\varepsilon)})\partial_{k}\psi_\alpha^{(\varepsilon)}-\epsilon_{ij}^{(\varepsilon)}\partial_{i}\varphi_\alpha^{(\varepsilon)}\partial_{j}\psi_\alpha^{(\varepsilon)})dx
    \end{align*}

Moreover, cell problems of HMM held above are equal to the following cell problems. Estimation of the macroscopic effective coefficients could be given directly by solving the microscopic equation systems as follows.
\begin{align}\label{hmm modone}
      \left\{\begin{array} {l@{\quad}l} \displaystyle
          -\frac{\partial}{\partial{x_{j}}}(c_{ijnh}^{(\varepsilon)}\frac{\partial{P_{n}^{kl}}}{\partial{x_{h}}}+e_{hij}^{(\varepsilon)}\frac{\partial{\Phi_{kl}}}{\partial{x_{h}}})=0\qquad
          &\mbox{ $x\in I_{\delta}(x_{\alpha})$}\\\\
          \displaystyle
          \frac{\partial}{\partial{x_{i}}}(e_{inh}^{(\varepsilon)}\frac{\partial{P_{n}^{kl}}}{\partial
            x_{h}}-\epsilon_{ih}^{(\varepsilon)}\frac{\partial{\Phi_{kl}}}{\partial{x_{h}}})=0\qquad
          &\mbox{ $x\in I_{\delta}(x_{\alpha})$}\\\\
          \displaystyle P_{k}^{kl}=x_{l} \qquad P_{n}^{kl}=0(n\neq
          k)\qquad&\mbox{ $x\in\partial I_{\delta}(x_{\alpha})$}\\\\
          \displaystyle \Phi_{kl}=0\qquad&\mbox{ $x\in\partial
            I_{\delta}(x_{\alpha})$}
        \end{array}\right.
    \end{align}
Estimation of $c_{ijkl}^H$ and $e_{ikl}^H$ are given by
\begin{eqnarray}\label{hmm coef1}
      c^{H}_{ijkl}(x_\alpha)&=&c_{ijnh}^{H}(x_\alpha)\langle\frac{\partial P_{n}^{kl}}{\partial
        x_{h}}\rangle_{I_{\delta}(x_\alpha)}=\langle c_{ijnh}^{(\varepsilon)}\frac{\partial{P_{n}^{kl}}}{\partial{x_{h}}}+e_{hij}^{(\varepsilon)}\frac{\partial{\Phi_{kl}}}{\partial{x_{h}}}\rangle_{I_{\delta}(x_\alpha)}\\
e^{H}_{ikl}(x_\alpha)&=&e_{inh}^{H}(x_\alpha)\langle\frac{\partial P_{n}^{kl}}{\partial
        x_{h}}\rangle_{I_{\delta}(x_\alpha)}=\langle e_{inh}^{(\varepsilon)}\frac{\partial{P_{n}^{kl}}}{\partial
        x_{h}}-\epsilon_{ih}^{(\varepsilon)}\frac{\partial{\Phi_{kl}}}{\partial{x_{h}}}\rangle_{I_{\delta}(x_\alpha)}
    \end{eqnarray}
and
   \begin{align}\label{hmm modtwo}
      \left\{\begin{array} {l@{\quad}l} \displaystyle
          -\frac{\partial}{\partial{x_{j}}}(c_{ijnh}^{(\varepsilon)}\frac{\partial{Q_{n}^{l}}}{\partial{x_{h}}}+e_{hij}^{(\varepsilon)}\frac{\partial{\Psi_{l}}}{\partial{x_{h}}})=0\qquad
          &\mbox{ $x\in I_{\delta}(x_{\alpha})$}\\\\
          \displaystyle
          \frac{\partial}{\partial{x_{i}}}(e_{inh}^{(\varepsilon)}\frac{\partial{Q_{n}^{l}}}{\partial
            x_{h}}-\epsilon_{ih}^{(\varepsilon)}\frac{\partial{\Psi_{l}}}{\partial{x_{h}}})=0\qquad
          &\mbox{ $x\in I_{\delta}(x_{\alpha})$}\\\\
          \displaystyle  Q_{n}^{l}=0\qquad&\mbox{ $x\in\partial I_{\delta}(x_{\alpha})$}\\\\
          \displaystyle \Psi_{l}=x_{l}\qquad&\mbox{ $x\in\partial
            I_{\delta}(x_{\alpha})$}
        \end{array}\right.
    \end{align}
Estimation of $\epsilon_{ij}$ and $e_{ikl}$ are given by
\begin{eqnarray}\label{hmm coef2}
      \epsilon_{il}^{H}(x_\alpha)&=&\epsilon_{ij}^{H}(x_\alpha)\langle\frac{\partial\Psi_{l}}{\partial
        x_{j}}\rangle_{I_{\delta}(x_\alpha)}=-\langle e_{inh}^{(\varepsilon)}\frac{\partial{Q_{n}^{l}}}{\partial
        x_{h}}-\epsilon_{ih}^{(\varepsilon)}\frac{\partial{\Psi_{l}}}{\partial{x_{h}}}\rangle_{I_{\delta}(x_\alpha)}\\
       e_{ikl}^{H}(x_\alpha)&=& e_{hkl}^{H}(x_\alpha)\langle\frac{\partial \Psi_{i}}{\partial
        x_{h}}\rangle_{I_{\delta}(x_\alpha)}=\langle c_{klnh}^{(\varepsilon)}\frac{\partial{Q_{n}^{i}}}{\partial{x_{h}}}+e_{hkl}^{(\varepsilon)}\frac{\partial{\Psi_{i}}}{\partial{x_{h}}}\rangle_{I_{\delta}(x_\alpha)}
    \end{eqnarray}

The following results assert that estimate of the effective coefficients $c^{H}_{ijkl}(x)$, $\epsilon_{il}^{H}(x)$ and $e_{ikl}^{H}(x)$
given above share the same symmetric property and elliptic property with $c_{ijkn}^{(\varepsilon)}(x)$, $\epsilon_{ij}^{(\varepsilon)}(x)$ and $e_{ijk}^{(\varepsilon)}(x)$ as follows.
\begin{lemma}\label{hmm lem1}
If $\mathbf{P}$, $\mathbf{Q}$, $\mathbf{\Phi}$, $\mathbf{\Psi}$ are solutions of (\ref{hmm modone}) and (\ref{hmm modtwo}), $e_{ikl}^H$, $c_{ijkl}^H$ and $\epsilon_{ij}^H$ defined by (\ref{hmm coef1}) and (\ref{hmm coef2}) have the following properties,
\begin{enumerate}
\item[(I)]
\begin{align*}
& e^{H}_{ikl}=\langle e_{inh}^{(\varepsilon)}\frac{\partial{P_{n}^{kl}}}{\partial
        x_{h}}-\epsilon_{ih}^{(\varepsilon)}\frac{\partial{\Phi_{kl}}}{\partial{x_{h}}}\rangle_{I_{\delta}}=\langle c_{klnh}^{(\varepsilon)}\frac{\partial{Q_{n}^{i}}}{\partial{x_{h}}}+e_{hkl}^{(\varepsilon)}\frac{\partial{\Psi_{i}}}{\partial{x_{h}}}\rangle_{I_{\delta}}
        \end{align*}
\item[(II)]
\begin{align*}
& c^{H}_{ijkl}=c^{H}_{jikl}=c^{H}_{klij},\quad e^H_{kij}=e^H_{kji},\quad\epsilon^H_{ij}=\epsilon^H_{ij}
\end{align*}
\item[(III)]
\begin{align*}
&\exists\alpha_{H}>0,s.t.c_{ijkl}^{H}X_{ij}X_{kl}\geq\alpha_{H}X_{ij}X_{kl}\quad and \quad\exists\beta_{H}>0,s.t.\epsilon_{ij}^{H}X_{i}X_{j}\geq\beta_{H}X_{i}X_{j}
\end{align*}
\end{enumerate}
\end{lemma}

\begin{proof}
The main idea for proving (I) and (II) is to take proper test functions in the variation forms of (\ref{hmm modone}) and (\ref{hmm modtwo}). The main idea for proving (III) is to make full use of property of energy functional of piezoelectric equation system, which is caused by the coupling of
mechanical field and electric field.\\
In (\ref{hmm modone}), taking
\begin{align*}
N_{n}^{kl}=P_{n}^{kl}\quad for\quad n\neq k\quad and\quad N_{k}^{kl}=P_{k}^{kl}-x_l
\end{align*}
we have,
\begin{align}\label{hmm lemone1}
      \left\{\begin{array} {l@{\quad}l} \displaystyle
          -\frac{\partial}{\partial{x_{j}}}(c_{ijnh}^{(\varepsilon)}\frac{\partial{N_{n}^{kl}}}{\partial{x_{h}}}+e_{hij}^{(\varepsilon)}\frac{\partial{\Phi_{kl}}}{\partial{x_{h}}})=-\frac{\partial{c_{ijkl}^{(\varepsilon)}}}{\partial{x_{j}}}\qquad
          &\mbox{ $x\in I_{\delta}(x_{\alpha})$}\\\\
          \displaystyle
          \frac{\partial}{\partial{x_{i}}}(e_{inh}^{(\varepsilon)}\frac{\partial{N_{n}^{kl}}}{\partial
            x_{h}}-\epsilon_{ih}^{(\varepsilon)}\frac{\partial{\Phi_{kl}}}{\partial{x_{h}}})=-\frac{\partial{e_{ikl}^{(\varepsilon)}}}{\partial{x_{i}}}\qquad
          &\mbox{ $x\in I_{\delta}(x_{\alpha})$}\\\\
          \displaystyle N_{n}^{kl}=0\qquad \Phi_{kl}=0\qquad&\mbox{ $x\in\partial
            I_{\delta}(x_{\alpha})$}
        \end{array}\right.
    \end{align}
The variational form of (\ref{hmm lemone1}) is,
\begin{align}\label{weak lemone1}
\left\{\begin{array} {l@{\quad}l} \displaystyle \int_{I_\delta(x_\alpha)}(c_{ijnh}^{(\varepsilon)}\frac{\partial{N_{n}^{kl}}}{\partial{x_{h}}}\frac{\partial{G_{i}}}{\partial{x_{j}}}+e_{hij}^{(\varepsilon)}\frac{\partial{\Phi_{kl}}}{\partial{x_{h}}}\frac{\partial{G_{i}}}{\partial{x_{j}}})dx
=\int_{I_\delta(x_\alpha)}c_{ijkl}^{(\varepsilon)}\frac{\partial{G_{i}}}{\partial{x_{j}}}dx\\\\
\displaystyle -\int_{I_\delta(x_\alpha)}(e_{inh}^{(\varepsilon)}\frac{\partial{N_{n}^{kl}}}{\partial
            x_{h}}\frac{\partial{\zeta}}{\partial{x_{i}}}-\epsilon_{ih}^{(\varepsilon)}\frac{\partial{\Phi_{kl}}}{\partial{x_{h}}}\frac{\partial{\zeta}}{\partial{x_{i}}})dx
            =\int_{I_\delta(x_\alpha)}e_{ikl}^{(\varepsilon)}\frac{\partial{\zeta}}{\partial{x_{i}}}dx\end{array}\right.
\end{align}
for any $\mathbf{G}\in H^1_0(I_\delta)\times H^1_0(I_\delta)\times H^1_0(I_\delta)$ and $\zeta\in H^1_0(I_\delta)$.\\
In (\ref{hmm modtwo}), taking $\psi_l=\Psi_l-x_l$, we have,
 \begin{align}\label{hmm lemone2}
      \left\{\begin{array} {l@{\quad}l} \displaystyle
          -\frac{\partial}{\partial{x_{j}}}(c_{ijnh}^{(\varepsilon)}\frac{\partial{Q_{n}^{l}}}{\partial{x_{h}}}+e_{hij}^{(\varepsilon)}\frac{\partial{\psi_{l}}}{\partial{x_{h}}})=-\frac{\partial e^{(\varepsilon)}_{lij} }{\partial x_j}\qquad
          &\mbox{ $x\in I_{\delta}(x_{\alpha})$}\\\\
          \displaystyle
          \frac{\partial}{\partial{x_{i}}}(e_{inh}^{(\varepsilon)}\frac{\partial{Q_{n}^{l}}}{\partial
            x_{h}}-\epsilon_{ih}^{(\varepsilon)}\frac{\partial{\psi_{l}}}{\partial{x_{h}}})=\frac{\partial\epsilon_{il}^{(\varepsilon)}}{\partial x_i}\qquad
          &\mbox{ $x\in I_{\delta}(x_{\alpha})$}\\\\
          \displaystyle  Q_{n}^{l}=0\qquad \psi_{l}=0\qquad&\mbox{ $x\in\partial
            I_{\delta}(x_{\alpha})$}
        \end{array}\right.
    \end{align}
The variational form of (\ref{hmm lemone2}) is,
\begin{align}\label{weak lemone2}
\left\{\begin{array} {l@{\quad}l}\displaystyle
\int_{I_\delta(x_\alpha)}(c_{ijnh}^{(\varepsilon)}\frac{\partial{Q_{n}^{l}}}{\partial{x_{h}}}\frac{\partial{F_{i}}}{\partial{x_{j}}}+e_{hij}^{(\varepsilon)}\frac{\partial{\psi_{l}}}{\partial{x_{h}}}\frac{\partial{F_{i}}}{\partial{x_{j}}})dx
=\int_{I_\delta(x_\alpha)}e^{(\varepsilon)}_{lij}\frac{\partial F_{i} }{\partial x_j}dx\\\\
\displaystyle\int_{I_\delta(x_\alpha)}(e_{inh}^{(\varepsilon)}\frac{\partial{Q_{n}^{l}}}{\partial
            x_{h}}\frac{\partial{\xi}}{\partial{x_{i}}}-\epsilon_{ih}^{(\varepsilon)}\frac{\partial{\psi_{l}}}{\partial{x_{h}}}\frac{\partial{\xi}}{\partial{x_{i}}})dx
=\int_{I_\delta(x_\alpha)}\epsilon_{il}^{(\varepsilon)}\frac{\partial\xi}{\partial x_i}
\end{array}\right.
\end{align}
for any $\mathbf{F}\in H^1_0(I_\delta(x_\alpha))\times H^1_0(I_\delta(x_\alpha))\times H^1_0(I_\delta(x_\alpha))$ and $\xi\in H^1_0(I_\delta(x_\alpha))$\\
(I) Taking $(\mathbf{G},\zeta)=(\mathbf{Q}^{i},\psi_{i})$ in (\ref{weak lemone1}), we have
\begin{align}\label{weak lemone12}
\left\{\begin{array} {l@{\quad}l} \displaystyle \int_{I_\delta(x_\alpha)}(c_{nhsm}^{(\varepsilon)}\frac{\partial{N_{s}^{kl}}}{\partial{x_{m}}}\frac{\partial{Q^i_{n}}}{\partial{x_{h}}}+e_{mnh}^{(\varepsilon)}\frac{\partial{\Phi_{kl}}}{\partial{x_{m}}}\frac{\partial{Q^i_{n}}}{\partial{x_{h}}})dx
=\int_{I_\delta(x_\alpha)}c_{nhkl}^{(\varepsilon)}\frac{\partial{Q^i_{n}}}{\partial{x_{h}}}dx\\\\
\displaystyle -\int_{I_\delta(x_\alpha)}(e_{hnm}^{(\varepsilon)}\frac{\partial{N_{n}^{kl}}}{\partial
            x_{m}}\frac{\partial{\psi_i}}{\partial{x_{h}}}-\epsilon_{hm}^{(\varepsilon)}\frac{\partial{\Phi_{kl}}}{\partial{x_{m}}}\frac{\partial{\psi_i}}{\partial{x_{h}}})dx
            =\int_{I_\delta(x_\alpha)}e_{ikl}^{(\varepsilon)}\frac{\partial{\psi_i}}{\partial{x_{i}}}dx\end{array}\right.
\end{align}
Taking $(\mathbf{F},\xi)=(\mathbf{N}^{kl},\Phi_{kl})$ in (\ref{weak lemone2}), we obtain
\begin{align}\label{weak lemone22}
\left\{\begin{array} {l@{\quad}l}\displaystyle
\int_{I_\delta(x_\alpha)}(c_{nhsm}^{(\varepsilon)}\frac{\partial{Q_{s}^{i}}}{\partial{x_{m}}}\frac{\partial{N^{kl}_{n}}}{\partial{x_{h}}}+e_{mnh}^{(\varepsilon)}\frac{\partial{\psi_{i}}}{\partial{x_{m}}}\frac{\partial{N^{kl}_{n}}}{\partial{x_{h}}})dx
=\int_{I_\delta(x_\alpha)}e^{(\varepsilon)}_{inh}\frac{\partial N^{kl}_{n} }{\partial x_h}dx\\\\
\displaystyle\int_{I_\delta(x_\alpha)}(e_{hsm}^{(\varepsilon)}\frac{\partial{Q_{s}^{i}}}{\partial
            x_{m}}\frac{\partial{\Phi_{kl}}}{\partial{x_{h}}}-\epsilon_{hm}^{(\varepsilon)}\frac{\partial{\psi_{i}}}{\partial{x_{m}}}\frac{\partial{\Phi_{kl}}}{\partial{x_{h}}})dx
=\int_{I_\delta(x_\alpha)}\epsilon_{ih}^{(\varepsilon)}\frac{\partial\Phi_{kl}}{\partial x_h}dx
\end{array}\right.
\end{align}
Since
\begin{align}
\langle e_{inh}^{(\varepsilon)}\frac{\partial{P_{n}^{kl}}}{\partial
        x_{h}}-\epsilon_{ih}^{(\varepsilon)}\frac{\partial{\Phi_{kl}}}{\partial{x_{h}}}\rangle_{I_{\delta}(x_\alpha)}=\langle e_{ikl}^{(\varepsilon)}+e_{inh}^{(\varepsilon)}\frac{\partial{N_{n}^{kl}}}{\partial
        x_{h}}-\epsilon_{ih}^{(\varepsilon)}\frac{\partial{\Phi_{kl}}}{\partial{x_{h}}}\rangle_{I_{\delta}(x_\alpha)}\\
\langle c_{klnh}^{(\varepsilon)}\frac{\partial{Q_{n}^{i}}}{\partial{x_{h}}}+e_{hkl}^{(\varepsilon)}\frac{\partial{\Psi_{i}}}{\partial{x_{h}}}\rangle_{I_{\delta}(x_\alpha)}
=\langle e_{ikl}^{(\varepsilon)}+c_{klnh}^{(\varepsilon)}\frac{\partial{Q_{n}^{i}}}{\partial{x_{h}}}+e_{hkl}^{(\varepsilon)}\frac{\partial{\psi_{i}}}{\partial{x_{h}}}\rangle_{I_{\delta}(x_\alpha)},
\end{align}
we get
\begin{align} e^{H}_{ikl}=\langle e_{inh}^{(\varepsilon)}\frac{\partial{P_{n}^{kl}}}{\partial
        x_{h}}-\epsilon_{ih}^{(\varepsilon)}\frac{\partial{\Phi_{kl}}}{\partial{x_{h}}}\rangle_{I_{\delta}}=\langle c_{klnh}^{(\varepsilon)}\frac{\partial{Q_{n}^{i}}}{\partial{x_{h}}}+e_{hkl}^{(\varepsilon)}\frac{\partial{\Psi_{i}}}{\partial{x_{h}}}\rangle_{I_{\delta}}
\end{align}
by the symmetric property of $c_{ijkn}^{(\varepsilon)}(x),\epsilon_{ij}^{(\varepsilon)}(x)$,$e_{ijk}^{(\varepsilon)}(x)$ and subtracting of (\ref{weak lemone12}) and (\ref{weak lemone22}) respectively.\\
(II) Taking $(\mathbf{G},\zeta)=(\mathbf{N}^{ij},\Phi_{ij})$ in (\ref{weak lemone1}) we obtain
\begin{align}
c_{ijkl}^{H}&=<c_{smnh}^{(\varepsilon)}\frac{\partial{P_{n}^{kl}}}{\partial{x_{h}}}\frac{\partial{P_{s}^{ij}}}{\partial{x_{m}}}+e_{hsm}^{(\varepsilon)}\frac{\partial{\Phi_{kl}}}{\partial{x_{h}}}\frac{\partial{P_{s}^{ij}}}{\partial{x_{m}}}>_{I_{\delta}(x_\alpha)}\\
&=<c_{smnh}^{(\varepsilon)}\frac{\partial{P_{n}^{kl}}}{\partial{x_{h}}}\frac{\partial{P_{s}^{ij}}}{\partial{x_{m}}}+\epsilon_{hm}^{(\varepsilon)}\frac{\partial{\Phi_{kl}}}{\partial{x_{h}}}\frac{\partial{\Phi_{ij}}}{\partial{x_{m}}}>_{I_{\delta}(x_\alpha)}
\end{align}
Therefore, $c_{ijkl}^{H}=c_{klij}^{H}$\\
Since
\begin{align}\label{hmm proofc}
c^H_{ijkl}=\langle c_{ijnh}^{(\varepsilon)}\frac{\partial{P_{n}^{kl}}}{\partial{x_{h}}}+e_{hij}^{(\varepsilon)}\frac{\partial{\Phi_{kl}}}{\partial{x_{h}}}\rangle_{I_{\delta}(x_\alpha)}
=\langle c_{ijkl}^{(\varepsilon)}+ c_{ijnh}^{(\varepsilon)}\frac{\partial{N_{n}^{kl}}}{\partial{x_{h}}}+e_{hij}^{(\varepsilon)}\frac{\partial{\Phi_{kl}}}{\partial{x_{h}}}\rangle_{I_{\delta}(x_\alpha)}
\end{align}
we have $c_{ijkl}^{H}=c_{jikl}^{H}=c_{ijlk}^{H}$ by the symmetric property of $c_{ijkn}^{(\varepsilon)}(x),\epsilon_{ij}^{(\varepsilon)}(x)$,$e_{ijk}^{(\varepsilon)}(x)$ in (\ref{hmm lemone1}).\\
Similarly, it is easy to prove that $e^H_{kij}=e^H_{kji}$, $\epsilon^H_{ij}=\epsilon^H_{ij}$\\
(III) By (\ref{hmm proofc}), for any $\mathbf{X}=(X_{ij})\neq \mathbf{0}$, we have
      \begin{align*}
       &&c_{ijkl}^{H}X_{ij}X_{kl}=<c_{ijkl}^{(\varepsilon)}X_{ij}X_{kl}+c_{ijnh}^{(\varepsilon)}\frac{\partial{N_{n}^{kl}}}{\partial{x_{h}}}X_{ij}X_{kl}+e_{hij}^{(\varepsilon)}\frac{\partial{\Phi_{kl}}}{\partial{x_{h}}}X_{ij}X_{kl}>_{I_{\delta}(x_\alpha)} \end{align*}
Setting $(\mathbf{N},\Phi)=(N_{n}^{kl}X_{kl},\Phi_{kl}X_{kl})$ in (\ref{weak lemone1}), we obtain
 \begin{align}\label{hmm pp}
   \left\{\begin{array} {l@{\quad}l} \displaystyle
   \int_{I_{\delta}(x_\alpha)}({c_{ijnh}^{(\varepsilon)}\frac{\partial{N_{n}}}{\partial{x_{h}}}+e_{hij}^{(\varepsilon)}\frac{\partial{\Phi}}{\partial{x_{h}}}})\frac{\partial G_{i}}{\partial x_{j}}=-\int_{I_{\delta}(x_\alpha)}c_{ijkl}^{(\varepsilon)}\frac{\partial G_{i}}{\partial x_{j}}X_{kl}\\\\
   \displaystyle
    \int_{I_{\delta}(x_\alpha)}({e_{inh}^{(\varepsilon)}\frac{\partial{N_{n}}}{\partial{x_{h}}}-\epsilon_{ih}^{(\varepsilon)}\frac{\partial{\Phi}}{\partial{x_{h}}}})\frac{\partial \zeta}{\partial x_{i}}=-\int_{I_{\delta}(x_\alpha)}e_{ikl}^{(\varepsilon)}\frac{\partial \zeta}{\partial x_{i}}X_{kl}
     \end{array}\right.
  \end{align}
Since $(\mathbf{N},\Phi)$ is the saddle point associated to the energy functional defined by
\begin{align}
   I[\mathbf{V},\Psi]=
   &\frac{1}{2}\int_{I_{\delta}(x_\alpha)}c_{ijnh}^{(\varepsilon)}(\frac{\partial V_{n}}{\partial x_{h}}+X_{nh})(\frac{\partial V_{i}}{\partial x_{j}}+X_{ij})dx
  -\frac{1}{2}\int_{I_{\delta}(x_\alpha)}\epsilon_{ih}^{(\varepsilon)}\frac{\partial \Psi}{\partial x_{i}}\frac{\partial \Psi}{\partial x_{h}}\nonumber\\
  &+\int_{I_{\delta}(x_\alpha)}e_{inh}^{(\varepsilon)}\frac{\partial \Psi}{\partial x_{i}}(\frac{\partial V_{n}}{\partial x_{h}}+X_{nh})
\end{align}
Thus, we obtain
\begin{eqnarray}
I[\mathbf{N},\Psi]\leq I[\mathbf{N},\Phi]\leq I[\mathbf{V},\Phi],\forall\mathbf{V}\in (H_{0}^{1}(I_{\delta}))^{3},\Psi\in H_{0}^{1}(I_{\delta})
\end{eqnarray}
which implies $I[\mathbf{N},\Phi]\geq I[\mathbf{N},0]>0$\\
Taking $(\mathbf{G},\zeta)=(\mathbf{N},\Phi)$ in (\ref{hmm pp}), we get
\begin{eqnarray}
c_{ijkl}^{H}X_{ij}X_{kl}=2I[\mathbf{N},\Phi]>0
\end{eqnarray}
$\mathbb{B}=\{X=(X_{ij}):X\mbox{ is symmetric and }X_{ij}X_{ij}=1\}$ and consider $\Psi:\mathbb{B}\mapsto R$ defined by $\Psi(X_{ij})=c_{ijkl}^{H}X_{ij}X_{kl}$, it is easy to see $\Psi$ is continuous on $\mathbb{B}$ and $\Psi >0$, which implies there exists $\alpha_{H}>0$ such that
\begin{eqnarray}
\Psi(\frac{X_{ij}}{\parallel X\parallel})\geq \alpha_{H},\mbox{ for any }\mathbf{ X}=(X_{ij})\neq\mathbf{0}
 \end{eqnarray}
Therefore, there exists $\alpha_{H}>0$ such that $c_{ijkl}^{H}X_{ij}X_{kl}\geq\alpha_{H}X_{ij}X_{kl}$\\
Similarly, it is easy to prove that there exists $\beta_{H}>0$ such that $\epsilon_{ij}^{H}X_{i}X_{j}\geq\beta_{H}X_{i}X_{j}$.\\
This completes the proof of Lemma \ref{hmm lem1}.
\end{proof}

In Section 5, we will give out both error
estimate between the HMM solutions and the homogenization solutions and error estimate of the effective coefficients for piezoelectric composite materials with periodic microstructure.

\section{Homogenization theory}
For those piezoelectric composite materials with periodic microstructure we consider the following equation system,
\begin{align}\label{micro equperiodic}
            \left\{\begin{array} {l@{\quad}l} \displaystyle
                -\frac{\partial{\sigma_{ij}^{(\varepsilon)}}}{\partial{x_{j}}}=f_{i}\qquad
                &\mbox{ ${x}\in\Omega\subset R^{d}$}\\
                \displaystyle    \frac{\partial{D_{i}^{(\varepsilon)}}}{\partial{x_{i}}}={0}\qquad&\mbox{ ${x}\in \Omega\subset R^{d}$}\\
                \displaystyle    \sigma_{ij}^{(\varepsilon)}=c_{ijkn}^{(\varepsilon)}\frac{\partial{u_{k}^{(\varepsilon)}}}{\partial{x_{n}}}+e_{kij}^{(\varepsilon)}\frac{\partial{\varphi^{(\varepsilon)}}}{\partial{x_{k}}}\\
                \displaystyle
                D_{i}^{(\varepsilon)}=e_{ikn}^{(\varepsilon)}\frac{\partial{u_{k}^{(\varepsilon)}}}{\partial{x_{n}}}-\epsilon_{ij}^{(\varepsilon)}\frac{\partial{\varphi^{(\varepsilon)}}}{\partial{x_{j}}}\\
            \displaystyle
            u_{i}^{(\varepsilon)}=0\qquad
            \varphi^{(\varepsilon)}=0\qquad&\mbox{  $x\in\partial\Omega$}\end{array}\right.
\end{align}
where $c_{ijkn}^{(\varepsilon)}(x)=c_{ijkn}(x,y)$, $e_{ijk}^{(\varepsilon)}(x)=e_{ijk}(x,y)$, $\epsilon_{ij}^{(\varepsilon)}(x)=\epsilon_{ij}(x,y)$ are periodic in $y=\frac{x}{\varepsilon}$ with reference cell $Y=[0,1]^d$ as their periodic. For understanding easily, we assume that the coefficients are periodic globally, i.e.
\begin{equation}\label{periodic cond}
c_{ijkn}^{(\varepsilon)}(x)=c_{ijkn}(\frac{x}{\varepsilon})\qquad e_{ijk}^{(\varepsilon)}(x)=e_{ijk}(\frac{x}{\varepsilon})\qquad\epsilon_{ij}^{(\varepsilon)}(x)=\epsilon_{ij}(\frac{x}{\varepsilon})
\end{equation}
Following the step in homogenization theory by \citet{Oleinik},
we get the homogenized equation system of (\ref{micro equperiodic}) as follows,
 \begin{align}\label{macro equperiodic}
            \left\{\begin{array} {l@{\quad}l} \displaystyle -\frac{\partial{\sigma_{ij}^{(0)}}}{\partial{x_{j}}}=f_{i}\qquad&\mbox{ $x\in\Omega$}\\
                \displaystyle \frac{\partial{D_{i}^{(0)}}}{\partial{x_{i}}}={0}\qquad&\mbox{ $x\in \Omega$}\\
                \displaystyle \sigma_{ij}^{(0)}=c_{ijkn}^{(0)}\frac{\partial{u_{k}^{(0)}}}{\partial{x_{n}}}+e_{kij}^{(0)}\frac{\partial{\varphi^{(0)}}}{\partial{x_{k}}}\\
                \displaystyle D_{i}^{(0)}=e_{ikn}^{(0)}\frac{\partial{u_{k}^{(0)}}}{\partial{x_{n}}}-\epsilon_{ij}^{(0)}\frac{\partial{\varphi^{(0)}}}{\partial{x_{j}}}\\
                \displaystyle  u_{i}^{(0)}=0��\qquad
            \varphi^{(0)}=0\qquad&\mbox{ $x\in \partial\Omega$} \end{array}\right.
          \end{align}
where $c_{ijkl}^{(0)}$, $e_{lij}^{(0)}$ and $\epsilon_{ij}^{(0)}$ are the homogenized effective coefficients given by,
 \begin{eqnarray}\label{asy coefficient}
            c_{ijkl}^{(0)} & = & \langle c_{ijkl}(y)+\tau_{ij}^{kl}(y)\rangle_{Y} \\
            e_{lij}^{(0)} & = &\langle e_{lij}(y)+\xi_{ij}^{l}(y)\rangle_{Y}\nonumber\\
                      & = &\langle e_{lij}(y)+d_{l}^{ij}(y)\rangle_{Y}\\
            \epsilon_{ij}^{(0)} & = & \langle\epsilon_{ij}(y)-g_{ij}(y)\rangle_{Y}
  \end{eqnarray}
where $\tau_{ij}^{kl}(y)$, $\xi_{ij}^{l}(y)$, $d_{l}^{ij}(y)$, $g_{ij}(y)$ can be got by solving the following two cell probelms,
  \begin{align}\label{asy cella}
          \left\{\begin{array} {l@{\quad}l} \displaystyle
              \frac{\partial{\tau_{ij}^{kl}}}{\partial{y_{j}}}=-\frac{\partial{c_{ijkl}(y)}}{\partial{y_{j}}}&\mbox{ $y\in Y$}\\
              \displaystyle    \frac{\partial{d_{i}^{kl}}}{\partial{y_{i}}}=-\frac{\partial{e_{ikl}(y)}}{\partial{y_{i}}}&\mbox{ $y\in Y$}\\
              \displaystyle    \tau_{ij}^{kl}=c_{ijnh}(y)\frac{\partial{N_{n}^{kl}(y)}}{\partial{y_{h}}}+e_{hij}(y)\frac{\partial{\phi_{kl}(y)}}{\partial{y_{h}}}\\
              \displaystyle d_{i}^{kl}=e_{inh}(y)\frac{\partial{N_{n}^{kl}(y)}}{\partial{y_{h}}}-\epsilon_{ih}(y)\frac{\partial{\phi_{kl}(y)}}{\partial{y_{h}}}\\
            \displaystyle N_{n}^{kl}(y)\mbox{ is periodic on }R^{d}\mbox{ and }\int_{Y} N_{n}^{kl}(y)dy = 0\\
             \displaystyle \phi_{kl}(y)\mbox{ is periodic on }R^{d}\mbox{ and } \int_{Y} \phi_{kl}(y)dy = 0
            \end{array}\right.
         \end{align}
 \begin{align}\label{asy cellb}
          \left\{\begin{array} {l@{\quad}l} \displaystyle \frac{\partial{\xi_{ij}^{l}}}{\partial{y_{j}}}=-\frac{\partial{e_{lij}(y)}}{\partial{y_{j}}}&\mbox{ $y\in Y$}\\
              \displaystyle    \frac{\partial{g_{i}^{l}}}{\partial{y_{i}}}=\frac{\partial{\epsilon_{il}(y)}}{\partial{y_{i}}}&\mbox{ $y\in Y$}\\
              \displaystyle    \xi_{ij}^{l}=c_{ijkh}(y)\frac{\partial{W_{k}^{l}(y)}}{\partial{y_{h}}}+e_{hij}(y)\frac{\partial{\psi_{l}(y)}}{\partial{y_{h}}}\\
              \displaystyle
              g_{il}=e_{ikh}(y)\frac{\partial{W_{k}^{l}(y)}}{\partial{y_{h}}}-\epsilon_{ih}(y)\frac{\partial{\psi_{l}(y)}}{\partial{y_{h}}}\\
             \displaystyle W_{k}^{l}(y)\mbox{ is periodic on }R^{d}\mbox{ and } \int_{Y} W_{k}^{l}(y)dy = 0\\
              \displaystyle \psi_{l}(y)\mbox{ is periodic on }R^{d}\mbox{ and } \int_{Y} \psi_{l}(y)dy = 0
            \end{array}\right.
          \end{align}

\begin{remark}
we may deduce that $\langle e_{lij}(y)+\xi_{ij}^{l}(y)\rangle_{Y}=\langle e_{lij}(y)+d_{l}^{ij}(y)\rangle_{Y}$ and $ c^{(0)}_{ijkl}=c^{(0)}_{jikl}=c^{(0)}_{klij},e^{(0)}_{kij}=e^{(0)}_{kji},\epsilon^{(0)}_{ij}=\epsilon^{(0)}_{ij}$ by the techniques we use in the proof of Lemma \ref{hmm lem1}
\end{remark}

So far we have derived the homogenization equations of piezoelectric composite materials with periodic microstructure in formal by asymptotic expansion. The following Lemmas and Theorems asserts the convergence of asymptotic expansion. The difficulty for the theoretical analysis of this part lays on the point that the regularity of the solutions of cell problems (\ref{asy cella}) and (\ref{asy cellb}) cannot reach $W^{1,\infty}(R^d)$, which implies traditional treatment in  second order elliptic equations in \citet{Oleinik} does not work.
\\

The following two lemmas play an important role in getting the convergence conclusion of asymptotic expansion. The original idea of them is from \citet{Suslina}. Our result is the promotion of Corollary 8.3 in \citet{Suslina} in piezoelectric equation system.
\begin{lemma}\label{hom lemone}
If $N_n^{kl}\in H^1_{per}(Y)\cap L^{\infty}(R^d)$ and $\phi_{kl}\in H^1_{per}(Y)\cap L^{\infty}(R^d)$,
  the operator $[\nabla  \mathbf{N}^{kl}]$ of multiplication by the matrix $[\nabla
   \mathbf{N}^{kl}]$ and the operator $[\nabla\phi_{kl}]$ of multiplication by the column $[\nabla \bf
   N^{kl}]$ are continuous from
   $H^{1}(R^d)$ to $\Re=L_{2}(R^d;C^d)$ and
   \begin{eqnarray*}
   \parallel [\nabla \mathbf{ N}^{kl}] \parallel_{H^{1}(R^d)\rightarrow \Re
   }+\parallel [\nabla\phi_{kl}] \parallel_{H^{1}(R^d)\rightarrow \Re
   }\leq C_{1}
   \end{eqnarray*}
   where $C_{1}$ depends on the coefficient of the material only.
\end{lemma}
\begin{lemma}\label{hom lemtwo}
If $W_k^{l}\in H^1_{per}(Y)\cap L^{\infty}(R^d)$ and $\psi_{l}\in H^1_{per}(Y)\cap L^{\infty}(R^d)$,
 the operator $[\nabla \mathbf{W}^{l}]$ of multiplication by the matrix $[\nabla
   \mathbf{W}^{l}]$ and the operator $[\nabla\Psi_{l}]$ of multiplication by the column $[\nabla\psi_{l}]$ are continuous from
   $H^{1}(R^d)$ to $\Re=L_{2}(R^d;C^d)$ and
   \begin{eqnarray*}
   \parallel [\nabla \mathbf{W}^{l}] \parallel_{H^{1}(R^3)\rightarrow \Re
   }+\parallel [\nabla\psi_{l}] \parallel_{H^{1}(R^3)\rightarrow \Re
   }\leq C_{2}
   \end{eqnarray*}
   where $C_{2}$ depends on the coefficient of the material only.
\end{lemma}
The following Lemma is the main result that leads the convergence conclusion of asymptotic expansion. For expression convenience, we introduce
\begin{align}\label{asy orderone}
    \displaystyle u_{k(1)}^{(\varepsilon)}&=u_{k}^{(0)}(x)+\varepsilon( N_n^{kl}(y)\frac{\partial u_k^{(0)}(x)}{\partial x_l}+W_n^{l}(y)\frac{\partial \varphi^{(0)}(x)}{\partial x_l})\\
     \displaystyle \varphi_{(1)}^{(\varepsilon)}&=\varphi^{(0)}(x)+\varepsilon (\phi_{kl}(y)\frac{\partial u_k^{(0)}(x)}{\partial x_l}+\psi_{l}(y)\frac{\partial \varphi^{(0)}(x)}{\partial x_l})
\end{align}
\begin{lemma}\label{hom lemthree}
If $u_{i}^{(0)}\in C^{2}(\Omega)$, $\varphi^{(0)}\in C^{2}(\Omega)$,
        $N_{n}^{kl}\in H^1_{per}(Y)\cap L^{\infty}(R^d)$, $W_{n}^{l}\in H^1_{per}(Y)\cap
        L^{\infty}(R^d)$, $\varphi_{kl}\in H^1_{per}(Y)\cap L^{\infty}(R^d)$, $\psi_{l}\in
        H^1_{per}(Y)\cap L^{\infty}(R^d)$,
 we have
  \begin{eqnarray*}
        \parallel
        u_{i(1)}^{(\varepsilon)}-u_{i}^{(\varepsilon)}\parallel_{H^{1}(\Omega)}+\parallel
        \varphi_{(1)}^{(\varepsilon)}-\varphi^{(\varepsilon)}\parallel_{H^{1}(\Omega)}\leq
        C\varepsilon^{\frac{1}{2}}
      \end{eqnarray*}
\end{lemma}
\begin{proof}
Take a truncation function $\tau^{\varepsilon}(x)$ satisfying,
\begin{eqnarray*}
   && 1.\quad\tau^{\varepsilon}(x)\in C_{0}^{\infty}(\Omega),\quad0\leq \tau^{\varepsilon}(x)\leq 1\quad and \quad \tau^{\varepsilon}(x)=1,\quad if\quad \rho(x,\partial\Omega)>\varepsilon\\
   && 2.\quad\varepsilon|\bigtriangledown\tau^{\varepsilon}(x)|\leq C ,\quad in\quad \Omega
    \end{eqnarray*}
Set
     \begin{eqnarray*}
     W_{k(1)}^{(\varepsilon)}&=&u_{k(1)}^{(\varepsilon)}-\varepsilon(1-\tau^{\varepsilon}(x))(N_n^{kl}\frac{\partial{u_{n}^{(0)}}}{\partial x_{l}}+W_{k}^{l}\frac{\partial\varphi^{(0)}}{\partial x_{l}})\\
    W_{(1)}^{(\varepsilon)}&=&\varphi_{(1)}^{(\varepsilon)}-\varepsilon(1-\tau^{\varepsilon}(x))(\phi_{kl}\frac{\partial{u_{n}^{(0)}}}{\partial x_{l}}+\psi_{l}\frac{\partial\varphi^{(0)}}{\partial x_{l}})
    \end{eqnarray*}
To get the conclusion, we separate the proof into two steps,
 \begin{eqnarray*}
    &&Step 1.\quad Prove \quad \parallel u_{k(1)}^{(\varepsilon)}-W_{k(1)}^{(\varepsilon)}\parallel_{H^{1}(\Omega)} \leq C{\varepsilon}^{\frac{1}{2}},\quad\parallel\varphi_{(1)}^{(\varepsilon)}-W_{(1)}^{(\varepsilon)}\parallel_{H^{1}(\Omega)}\leq C{\varepsilon}^{\frac{1}{2}}\\
    &&Step 2.\quad Prove \quad \parallel u_{k}^{(\varepsilon)}-W_{k(1)}^{(\varepsilon)}\parallel_{H^{1}(\Omega)} \leq C\varepsilon^{\frac{1}{2}},\quad\parallel\varphi^{(\varepsilon)}-W_{(1)}^{(\varepsilon)}\parallel_{H^{1}(\Omega)}\leq C\varepsilon^{\frac{1}{2}}
    \end{eqnarray*}
For step 1,
 \begin{align*}
    &\parallel u_{k(1)}^{(\varepsilon)}-W_{k(1)}^{(\varepsilon)}\parallel^{2}_{H^{1}(\Omega)}\\
    =&\int_{\Omega}[\varepsilon(1-\tau^{\varepsilon}(x))(N_n^{kl}\frac{\partial{u_{n}^{(0)}}}{\partial x_{l}}+W_{k}^{l}\frac{\partial\varphi^{(0)}}{\partial x_{l}})]^{2}dx\\
    &+\sum_{j}\int_{\Omega}[\varepsilon\frac{\partial\tau^{\varepsilon}(x)}{\partial x_{j}}(N_n^{kl}\frac{\partial{u_{n}^{(0)}}}{\partial x_{l}}+W_{k}^{l}\frac{\partial\varphi^{(0)}}{\partial x_{l}})]^{2}dx\\
    &+\sum_{j}\int_{\Omega}[(1-\tau^{\varepsilon}(x))(\frac{\partial N_{k}^{ml}}{\partial y_j}\frac{\partial{u_{m}^{(0)}}}{\partial x_{l}}+\frac{\partial W_{k}^{l}}{\partial y_j}\frac{\partial\varphi^{(0)}}{\partial x_{l}})]^{2}dx\\
    &+\sum_{j}\int_{\Omega}[\varepsilon(1-\tau^{\varepsilon}(x))(N_{k}^{ml}\frac{\partial^{2}u_{m}^{(0)}}{\partial x_j \partial x_l}+W_{k}^{l}\frac{\partial^{2}\varphi^{(0)}}{\partial x_j \partial x_l})]^{2}dx\\
    =&(I)+(II)+(III)+(IV)
    \end{align*}
    For (I), (II), (IV), there is no essential difference from the analysis of corresponding parts in second order elliptic equations in \citet{Oleinik}. So it is easy to get
    \begin{eqnarray*}
    (I)+(II)+(IV)\leq C_{1}(\varepsilon^{2}+\varepsilon \rho^{d-1})(\parallel u_{i}^{(0)}\parallel^2_{H^2(\Omega)}+\parallel \varphi^{(0)}\parallel^2_{H^2(\Omega)})
    \end{eqnarray*}
    For (III), we notice that the regularity of the solutions of cell
    problems (\ref{asy cella})and (\ref{asy cellb}) cannot reach $W^{1,\infty}(R^d)$, which makes the traditional treatment in  second order elliptic equations in \citet{Oleinik} not work. In this case, we lent help from Lemma \ref{hom lemone} and Lemma \ref{hom lemtwo}. We take the first term of (III) as example to give out proof in details. The proof of other terms are similar.\\
    By lemma \ref{hom lemone},we have
    \begin{align}\label{imp}
    \int_{R^d}[\frac{\partial N_{n}^{kl}}{\partial y_j}\frac{\partial N_{n}^{kl}}{\partial y_j}]\mid v(y)\mid^{2}dy\leq C(\int_{R^d}\mid v\mid^2 dy+\int_{R^d}\mid \nabla v\mid^2 dy)\qquad \forall v\in H^{1}(\Omega)
     \end{align}
     which implies
 \begin{eqnarray*}
     && \int_{R^d}[\frac{\partial N_{n}^{kl}}{\partial y_j}\frac{\partial N_{n}^{kl}}{\partial y_j}]\mid w(x)\mid^{2}dx\\
    &=&\varepsilon^{d}\int_{R^d}[\frac{\partial N_{n}^{kl}}{\partial x_j}\frac{\partial N_{n}^{kl}}{\partial x_j}]\mid w( \varepsilon x)\mid^{2}dx\\
    &\leq & \int_{R^d}\mid w(\varepsilon x)\mid^2 dx+\int_{R^d}\mid \nabla_{x} w(\varepsilon x)\mid^2 dx \\
     &=&\int_{R^d}\mid w(x)\mid^2+\varepsilon^{2}\int_{R^d}\mid\nabla_{x}w(x)\mid^2 dx
    \end{eqnarray*}
 Taking $w(x)=(1-\tau^{\varepsilon}(x))\frac{\partial u_{k}^{(0)}}{\partial x_{l}}$, we obtain
 \begin{eqnarray*}
    &&\int_{\Omega}[(1-\tau^{\varepsilon}(x))(\frac{\partial N_{n}^{kl}}{\partial y_j}\frac{\partial{u_{k}^{(0)}}}{\partial x_{l}})]^2 dx\\
   &\leq &\int_{R^3}[(1-\tau^{\varepsilon}(x))(\frac{\partial N_{n}^{kl}}{\partial y_j}\frac{\partial{u_{k}^{(0)}}}{\partial x_{l}})]^2 dx\\
   &\leq &\int_{R^3}\mid (1-\tau^{\varepsilon}(x))\frac{\partial u_{k}^{0}}{\partial x_{l}}\mid^{2} dx+\varepsilon^{2}\int_{R^3}\mid \nabla((1-\tau^{\varepsilon}(x))\frac{\partial u_{k}^{0}}{\partial x_{l}})\mid^{2} dx\\
   &=&\int_{Q^\varepsilon}\mid (1-\tau^{\varepsilon}(x))\frac{\partial u_{k}^{0}}{\partial x_{l}}\mid^{2}dx+\varepsilon^2\int_{Q^\varepsilon}\mid \nabla((1-\tau^{\varepsilon}(x))\frac{\partial u_{k}^{0}}{\partial x_{l}})\mid^{2} dx\\
   &\leq&C\varepsilon\rho^{d-1}
    \end{eqnarray*}
   Therefore,
     \begin{eqnarray*}
    (I)+(II)+(III)+(IV)\leq C_{1}(\varepsilon^{2}+\varepsilon \rho^{d-1})(\parallel u_{i}^{(0)}\parallel^2_{H^2(\Omega)}+\parallel \varphi^{(0)}\parallel^2_{H^2(\Omega)})
    \end{eqnarray*}
   Similarly, we can get $\parallel\varphi_{(1)}^{(\varepsilon)}-W_{(1)}^{(\varepsilon)}\parallel_{H^{1}(\Omega)}\leq C{\varepsilon}^{\frac{1}{2}}$\\
 This completes the proof of Step 1.\\
For Step 2, there is no essential difference from the analysis of corresponding parts in second order elliptic equations in \citet{Oleinik}. Therefore, we omit the proof of this part.
Thus, it is a direct corollary from Step 1 and Step 2 that
\begin{eqnarray*}
        \parallel
        u_{i(1)}^{(\varepsilon)}-u_{i}^{(\varepsilon)}\parallel_{H^{1}(\Omega)}+\parallel
        \varphi_{(1)}^{(\varepsilon)}-\varphi^{(\varepsilon)}\parallel_{H^{1}(\Omega)}\leq
        C\varepsilon^{\frac{1}{2}}
      \end{eqnarray*}
This completes the proof of Lemma \ref{hom lemthree}.
\end{proof}
Thus, the convergence conclusion of asymptotic expansion is got in the sense of H convergence as follows.
\begin{theorem}\label{hom thm}
If $u_{i}^{(0)}\in C^{2}(\Omega)$, $\varphi^{(0)}\in C^{2}(\Omega)$,
        $N_{n}^{kl}\in H^1_{per}(Y)\cap L^{\infty}(R^d)$, $W_{n}^{l}\in H^1_{per}(Y)\cap
        L^{\infty}(R^d)$, $\varphi_{kl}\in H^1_{per}(Y)\cap L^{\infty}(R^d)$, $\psi_{l}\in
        H^1_{per}(Y)\cap L^{\infty}(R^d)$,
we have
      \begin{eqnarray*}
       u_{i}^{(\varepsilon)} \rightharpoonup u_{i}^{(0)},\quad \varphi^{(\varepsilon)} \rightharpoonup \varphi^{(0)} \quad in \quad H_{0}^{1}(\Omega) \quad as \quad \varepsilon  \to 0\\
      \sigma_{ij}^{(\varepsilon)} \rightharpoonup \sigma_{ij}^{(0)},\quad D_{i}^{(\varepsilon)} \rightharpoonup D_{i}^{(0)} \quad in \quad L^{2}(\Omega) \quad as \quad \varepsilon  \to  0
      \end{eqnarray*}
\end{theorem}
\begin{proof}
The property of the mean value yields
\begin{eqnarray}\label{asy condone}
     u_{i(1)}^{(\varepsilon)} \rightharpoonup u_{i}^{(0)},\quad \varphi_{(1)}^{(\varepsilon)} \rightharpoonup \varphi^{(0)} \quad in \quad H^{1}(\Omega) \quad as \quad \varepsilon \to 0\\
     \sigma_{ij}^{(\varepsilon)} \rightharpoonup \sigma_{ij}^{(0)},\quad D_{i}^{(\varepsilon)} \rightharpoonup D_{i}^{(0)} \quad in \quad L^{2}(\Omega) \quad as \quad \varepsilon  \to  0
 \end{eqnarray}
following the steps in \citet{Oleinik}.\\
Then, by the conclusions in Lemma \ref{hom lemthree}, we have
 \begin{eqnarray}\label{asy condtwo}
     u_{i(1)}^{(\varepsilon)} \rightarrow u_{i}^{(\varepsilon)},\quad\varphi_{(1)}^{(\varepsilon)} \rightarrow \varphi^{(\varepsilon)}\quad in \quad H^{1}(\Omega) \quad as \quad \varepsilon  \to  0
 \end{eqnarray}
Thus, it is obviously that
 \begin{eqnarray*}
       u_{i}^{(\varepsilon)} \rightharpoonup u_{i}^{(0)}, \varphi^{(\varepsilon)} \rightharpoonup \varphi^{(0)} \quad in \quad H_{0}^{1}(\Omega) \quad as \quad \varepsilon  \to  0\\
      \sigma_{ij}^{(\varepsilon)} \rightharpoonup \sigma_{ij}^{(0)}, D_{i}^{(\varepsilon)} \rightharpoonup D_{i}^{(0)} \quad in \quad L^{2}(\Omega) \quad as \quad \varepsilon  \to  0
      \end{eqnarray*}
by (\ref{asy condone}) and (\ref{asy condtwo}). This completes the proof of Theorem \ref{hom thm}.
\end{proof}
\begin{remark}
It is not necessary to assume the solutions of the cell problems (\ref{asy cella}) and (\ref{asy cellb}) belong to $L^\infty(R^d)$ to certify the convergence of asymptotic expansion. However, to get the order of $\varepsilon$ in Lemma \ref{hom lemthree}, the conditions $ N_{n}^{kl}\in H^1_{per}(Y)\cap L^{\infty}(R^d)$, $W_{n}^{l}\in H^1_{per}(Y)\cap L^{\infty}(R^d)$, $\varphi_{kl}\in H^1_{per}(Y)\cap L^{\infty}(R^d)$, $\psi_{l}\in
        H^1_{per}(Y)\cap L^{\infty}(R^d)$ should be involved.
\end{remark}
So far, for those materials with periodic microstructure, the homogenized equation have been deduced from the microscopic
model of piezoelectric composite material under the frame of homogenization theory
and the corresponding convergence
analysis is complected.
\section{Error estimate of multi-scale modeling for piezoelectric composite materials}
In Section 3, we have designed multi-scale model for piezoelectric composite materials under the framework of HMM. However, HMM is an algorithm estimating the missing macro-scale data by solving the fine scale problem locally. Therefore, both error
estimate between HMM solutions and homogenization solutions and error estimate of the effective coefficients will be held in the following part. We restrict our discussion for these analysis to the periodic case, i.e., we assume that (\ref{periodic cond}) holds.
\begin{theorem}\label{hmm thm1}
Denote the solution of (\ref{macro equperiodic}) and the HMM solution by $(\mathbf{u}^{(0)},\varphi^{(0)})$ and $(\mathbf{u}^H,\varphi^{H})$, respectively. Let
 \begin{eqnarray*}
        &&e_{c}(HMM)=\max_{\substack { x_l\in K,K\in\mathscr{T}_{H}}}\parallel
        c_{ijkl}^{(0)}-c_{ijkl}^{H}\parallel\\
        &&e_{e}(HMM)=\max_{\substack {x_l\in K,K\in \mathscr{T}_{H}}}\parallel
        e_{ijk}^{(0)}-e_{ijk}^{H}\parallel\\
        &&e_{d}(HMM)=\max_{\substack {x_l\in K,K\in \mathscr{T}_{H}}}\parallel
        \epsilon_{ij}^{(0)}-e_{ijk}^{H}\parallel
      \end{eqnarray*}
 where $\parallel\cdot\parallel$ is the Euclidean norm. If $(\mathbf{u}^{(0)}$, $\varphi^{(0)})$ are sufficiently smooth, then there exists a constant C independent of $\varepsilon,\delta,H$, such that
  \begin{eqnarray*}
        \parallel u_{i}^{(0)}-u_{i}^{H}\parallel_{1}+\parallel
        \varphi^{(0)}-\varphi^{H}\parallel_{1}\quad\leq\quad
        C(H+e_{c}(HMM)+e_{e}(HMM)+e_{d}(HMM))
      \end{eqnarray*}
\end{theorem}
\begin{proof}
It is obviously that
\begin{eqnarray}
&&\parallel u_{i}^{(0)}-u_{i}^{H}\parallel_{H^1(\Omega)}\quad\leq\quad \parallel u_{i}^{(0)}-\Pi u_{i}^{(0)}\parallel_{H^1(\Omega)}+\parallel \Pi u_{i}^{(0)}-u_{i}^{H}\parallel_{H^1(\Omega)}\\
&&\parallel\varphi^{(0)}-\varphi^{H}\parallel_{H^1(\Omega)}\quad\leq\quad\parallel\varphi-\Pi\varphi^{(0)}\parallel_{H^1(\Omega)}+\parallel\Pi\varphi^{(0)}-\varphi^{H}\parallel_{H^1(\Omega)}
\end{eqnarray}
where $\Pi$ is 1st-order Lagrange interpolate operator.\\
Using interpolation error estimate in \citet{FEM2}, we have
\begin{eqnarray}
&&\parallel u_{i}^{(0)}-\Pi u_{i}^{(0)}\parallel_{H^1(\Omega)}\quad\leq\quad CH\mid  u_{i}^{(0)} \mid_2\\
&&\parallel\varphi^{(0)}-\Pi\varphi^{(0)}\parallel_{H^1(\Omega)}\quad\leq\quad CH\mid \varphi^{(0)} \mid_2
\end{eqnarray}
Then, it remains analysis of $\parallel \Pi u_{i}^{(0)}-u_{i}^{H}\parallel_{H^1(\Omega)}$ and $\parallel\varphi^{H}-\Pi\varphi^{(0)}\parallel_{H^1(\Omega)}$\\
By adding, (\ref{hmm macroweak}) yields
\begin{equation}\label{hmm proof2(2)}
c_{H}(\mathbf{u}^H,\mathbf{v}^H)+e_{H}(\mathbf{v}^H,\varphi^H)-e_{H}(\mathbf{u}^H,\psi^H)+d_{H}(\varphi^H,\psi^H)=(\mathbf{f},\mathbf{v}^H)
\end{equation}
For homogenized equations (\ref{macro equperiodic}), the variational form is
 \begin{align}\label{hmm proof2(3)}
      \left\{\begin{array} {l@{\quad}l} \displaystyle
          c(\mathbf{u}^{(0)},\mathbf{v}^{(0)})+e(\mathbf{v}^{(0)},\varphi^{(0)})=(\mathbf{f},\mathbf{v}^{(0)})\qquad&\forall\mathbf{v}^{(0)}\in
          H^{1}_{0}\times H^{1}_{0}\times H^{1}_{0}\\
          \displaystyle
          -e(\mathbf{u}^{(0)},\psi^{(0)})+d(\varphi^{(0)},\psi^{(0)})=0\qquad
          &\forall\psi^{(0)}\in H^{1}_{0}
        \end{array}\right.
    \end{align}
where
  \begin{eqnarray*}
      c(\mathbf{u}^{(0)},\mathbf{v}^{(0)})&=&\int_{\Omega}c_{ijkl}^{(0)}s_{ij}(\mathbf{u}^{(0)})s_{kl}(\mathbf{v}^{(0)})dx\\
      e(\mathbf{v}^{(0)},\varphi^{(0)})&=&\int_{\Omega}e_{ijk}^{(0)}s_{ij}(\mathbf{v}^{(0)})\partial_{k}\varphi^{(0)} dx\\
      d(\varphi^{(0)},\psi^{(0)})&=&\int_{\Omega}\epsilon_{ij}^{(0)}\partial_{i}\varphi^{(0)}\partial_{j}\psi^{(0)} dx
    \end{eqnarray*}
By adding and $X_H\subset H^{1}_{0}$, (\ref{hmm proof2(3)}) yields
\begin{eqnarray}\label{hmm proof2(4)}
&&c(\mathbf{u}^{(0)},\mathbf{v}^H)+e(\mathbf{v}^H,\varphi^{(0)})-e(\mathbf{u}^{(0)},\psi^H)+d(\varphi^{(0)},\psi^H)=(\mathbf{f},\mathbf{v}^H)
\end{eqnarray}
for any $\mathbf{v}^H\in X_{H}\times X_{H}\times X_{H}$ and $\psi^H\in X_{H}$
$(\ref{hmm proof2(4)})-(\ref{hmm proof2(2)})$, we have,
\begin{align}
&(c(\mathbf{u}^{(0)},\mathbf{v}^H)-c_{H}(\mathbf{u}^H,\mathbf{v}^H))+(d(\varphi^{(0)},\psi^H)-d_{H}(\varphi^H,\psi^H))\nonumber\\
+&(e(\mathbf{v}^H,\varphi^{(0)})-e_{H}(\mathbf{v}^H,\varphi^H))-(e(\mathbf{u}^{(0)},\psi^H)-e_{H}(\mathbf{u}^H,\psi^H))=0\label{hmm proof2(5)}
\end{align}
By (\ref{hmm proof2(5)}), it is obviously that
\begin{align}\label{hmm proof2(6)}
&(c(\mathbf{u}^{(0)},\mathbf{v}^H)-c_H(\mathbf{u}^{(0)},\mathbf{v}^H))+c_H(\mathbf{u}^{(0)}-\Pi\mathbf{u}^{(0)},\mathbf{v}^H)\nonumber\\
+&(d(\varphi^{(0)},\psi^H)-d_H(\varphi^{(0)},\psi^H))+d_{H}(\varphi^{(0)}-\Pi\varphi^{(0)},\psi^H))\nonumber\\
+&(e(\mathbf{v}^H,\varphi^{(0)})-e_H(\mathbf{v}^H,\varphi^{(0)}))+e_{H}(\mathbf{v}^H,\varphi^{(0)}-\Pi\varphi^{(0)}))\nonumber\\
-&(e(\mathbf{u}^{(0)},\psi^H)-e_H(\mathbf{u}^{(0)},\psi^H))-e_{H}(\mathbf{u}^{(0)}-\Pi\mathbf{u}^{(0)},\psi^H)\nonumber\\
=&c_H(\mathbf{u}^H-\Pi\mathbf{u}^{(0)},\mathbf{v}^H)+d_{H}(\varphi^H-\Pi\varphi^{(0)},\psi^H))+e_{H}(\mathbf{v}^H,\varphi^H-\Pi\varphi^{(0)}))-e_{H}(\mathbf{u}^H-\Pi\mathbf{u}^{(0)},\psi^H)
\end{align}
Setting $\mathbf{v}^H=\mathbf{u}^H-\Pi\mathbf{u}^{(0)}$,$\psi^H=\varphi^H-\Pi\varphi^{(0)}$ in (\ref{hmm proof2(6)}), we have
\begin{eqnarray*}
\mbox{The right hand side of }(\ref{hmm proof2(6)})\geq \alpha \parallel\mathbf{u}^H-\Pi\mathbf{u}^{(0)}\parallel^2_{1}+\beta\parallel\varphi^H-\Pi\varphi^{(0)}\parallel^2_{1}
\end{eqnarray*}
by Fridriches's inequality and Lemma \ref{hmm lem1}.\\
Using H$\ddot{o}$lder inequality, we obtain
\begin{align*}
&\mbox{The left hand side of }(\ref{hmm proof2(6)})\\
\leq & e_c(HMM)\mid \mathbf{u}^{(0)}\mid_1\mid\mathbf{u}^H-\Pi\mathbf{u}^{(0)}\mid_1+C_1\mid\mathbf{u}^H-\Pi\mathbf{u}^{(0)}\mid_1\mid\mathbf{u}^{(0)}-\Pi\mathbf{u}^{(0)}\mid_1\\
+&e_d(HMM)\mid\varphi^{(0)}\mid_1\mid\varphi^H-\Pi\varphi^{(0)}\mid_1+C_2\mid\varphi^{(0)}-\Pi\varphi^{(0)}\mid_1\mid\varphi^H-\Pi\varphi^{(0)}\mid_1\\
+&e_e(HMM)\mid\varphi^{(0)}\mid_1\mid\mathbf{u}^H-\Pi\mathbf{u}^{(0)}\mid_1+C_3\mid\varphi^{(0)}-\Pi\varphi^{(0)}\mid_1\mid\mathbf{u}^H-\Pi\mathbf{u}^{(0)}\mid_1\\
+&e_e(HMM)\mid \mathbf{u}^{(0)}\mid_1\mid\varphi^H-\Pi\varphi^{(0)}\mid_1+C_4\mid\varphi^H-\Pi\varphi^{(0)}\mid_1\mid\mathbf{u}^{(0)}-\Pi\mathbf{u}^{(0)}\mid_1
\end{align*}
Then, by Cauchy inequality and interpolation error estimate in \citet{FEM2}, we have
\begin{eqnarray*}
\parallel \Pi u_{i}^{(0)}-u_{i}^{H}\parallel_{H^1(\Omega)}+\parallel\varphi^{H}-\Pi\varphi^{(0)}\parallel_{H^1(\Omega)}\leq C(H+e_c(HMM)+e_d(HMM)+e_e(HMM))
\end{eqnarray*}
This completes the proof of Theorem \ref{hmm thm1}.
\end{proof}
\begin{remark}
At this stage, no assumption on the form of $c_{ijkl}^{(\varepsilon)}$, $e_{ijk}^{(\varepsilon)}$ and $\epsilon_{ij}^{(\varepsilon)}$ is needed. $(\mathbf{u}^{(0)},\varphi^{(0)})$ can be the
solution of an arbitrary macroscopic equation system with the same right-hand side as
in (\ref{micro equ})
\end{remark}

\begin{theorem}\label{hmmthm2}
For those piezoelectric composite materials with periodic microstructure, if $N_n^{kl}\in H^1_{per}(Y)\cap L^{\infty}(R^d)$ and $\phi_{kl}\in H^1_{per}(Y)\cap L^{\infty}(R^d)$, we have the following estimate
     \begin{eqnarray*}
        &&e_{c}(HMM)\leq C_{c}\frac{\varepsilon}{\delta}\\
        &&e_{e}(HMM)\leq C_{e}\frac{\varepsilon}{\delta}\\
        &&e_{d}(HMM)\leq C_{d}\frac{\varepsilon}{\delta}
      \end{eqnarray*}
where $C_{c}$, $C_{e}$, $C_{d}$ are constants independent of ~$\varepsilon$ and $\delta$.
\end{theorem}
\begin{proof}
By (\ref{hmm modone}) and (\ref{hmm coef1}), we obtain
\begin{align}
 c_{ijkl}^{H}&=\langle c_{smnh}^{(\varepsilon)}\frac{\partial{P_{n}^{kl}}}{\partial{x_{h}}}\frac{\partial{P_{s}^{ij}}}{\partial{x_{m}}}+e_{hsm}^{(\varepsilon)}\frac{\partial{\Phi_{kl}}}{\partial{x_{h}}}\frac{\partial{P_{s}^{ij}}}{\partial{x_{m}}}\rangle_{I_{\delta}}\label{hmm ela2}\\ &=\langle c_{smnh}^{(\varepsilon)}\frac{\partial{P_{n}^{kl}}}{\partial{x_{h}}}\frac{\partial{P_{s}^{ij}}}{\partial{x_{m}}}+\epsilon_{hm}^{(\varepsilon)}\frac{\partial{\Phi_{kl}}}{\partial{x_{h}}}\frac{\partial{\Phi_{ij}}}{\partial{x_{m}}}\rangle_{I_{\delta}}\label{hmm ela3}
    \end{align}
Similarly, by (\ref{asy cella}) and (\ref{asy coefficient}), we have
\begin{align}
c_{ijkl}^{(0)}&=\langle c_{smnh}(y)(\frac{\partial N_{s}^{ij}}{\partial y_{m}}+\delta_{is}\delta_{jm})(\frac{\partial N_{n}^{kl}}{\partial y_{h}}+\delta_{kn}\delta_{hl})+e_{hsm}\frac{\partial \phi_{kl}}{\partial y_{h}}(\frac{\partial N_{s}^{ij}}{\partial y_{m}}+\delta_{is}\delta_{jm})\rangle_{Y}\label{hmm ce2}\\
  &=\langle c_{smnh}(y)(\frac{\partial N_{s}^{ij}}{\partial y_{m}}+\delta_{is}\delta_{jm})(\frac{\partial N_{n}^{kl}}{\partial y_{h}}+\delta_{kn}\delta_{hl})+\epsilon_{hm}\frac{\partial \phi_{kl}}{\partial y_{h}}\frac{\partial \phi_{ij}}{\partial y_{m}}\rangle_{Y}\label{hmm ce3}
  \end{align}
The solution of (\ref{hmm modone}), $(\mathbf{P}^{kl},\Phi_{kl})$ has the following expansion,
\begin{eqnarray}
  &&P_{n}^{kl}={P_{n}^{kl}}^{(0)}+\varepsilon(N_n^{kj}(y)\frac{\partial {P_{n}^{kl}}^{(0)}(x)}{\partial x_j}+W_n^{j}(y)\frac{\partial \Phi_{kl}^{(0)}(x)}{\partial x_j})+\varepsilon\theta_{n}^{kl}\\
  &&\Phi_{kl}=\Phi_{kl}^{(0)}+\varepsilon(\phi_{nj}(y)\frac{\partial  {P_{n}^{kl}}^{(0)}(x)}{\partial x_j}+\psi_{j}(y)\frac{\partial  \Phi_{kl}^{(0)}(x)}{\partial x_j})+\varepsilon\Theta_{kl}
  \end{eqnarray}
where $({\mathbf{P}^{kl}}^{(0)},\Phi_{kl}^{(0)})$ is the solution of the following equation system
\begin{align}\label{hmm proof2(7)}
      \left\{\begin{array} {l@{\quad}l} \displaystyle
          -\frac{\partial}{\partial{x_{j}}}(c_{ijnh}^{(0)}\frac{\partial{{P_{n}^{kl}}^{(0)}}}{\partial{x_{h}}}+e_{hij}^{(0)}\frac{\partial{{\Phi_{kl}}^{(0)}}}{\partial{x_{h}}})=0\qquad
          &\mbox{
            $x\in I_{\delta}(x_{\alpha})$}\\\\
          \displaystyle
          \frac{\partial}{\partial{x_{i}}}(e_{inh}^{(0)}\frac{\partial{{P_{n}^{kl}}^{(0)}}}{\partial
            x_{h}}-\epsilon_{ih}^{(0)}\frac{\partial{{\Phi_{kl}}^{(0)}}}{\partial{x_{h}}})=0\qquad
          &\mbox{
            $x\in I_{\delta}(x_{\alpha})$}\\\\
          \displaystyle {P_{k}^{kl}}^{(0)}=x_{l} \qquad {P_{n}^{kl}}^{(0)}=0(n\neq
          k)\qquad&\mbox{
            $x\in\partial I_{\delta}(x_{\alpha})$}\\\\
          \displaystyle {\Phi_{kl}}^{(0)}=0\qquad&\mbox{ $x\in\partial
            I_{\delta}(x_{\alpha})$}
        \end{array}\right.
    \end{align}
We know that $c_{ijnh}^{(0)},e_{inh}^{(0)},\epsilon_{ih}^{(0)}$ are constants under the assumption (\ref{periodic cond}).
Thus, we can write out the solution of (\ref{hmm proof2(7)}) explicitly as follows by the existent and unique property of the solution of (\ref{hmm proof2(7)}).
 \begin{eqnarray}
&&{P_{k}^{kl}}^{(0)}=x_{l} \qquad {P_{n}^{kl}}^{(0)}=0(n\neq
          k)\qquad{\Phi_{kl}}^{(0)}=0\qquad\mbox{on $
            I_{\delta}(x_{\alpha})$}
  \end{eqnarray}
Then,
\begin{eqnarray}
\nabla P_{k}^{kl}&=&e_l+\nabla_y N_k^{kl}(y)+\nabla(\varepsilon\theta_{k}^{kl})=\mathbf{W_{k}^{kl}}+\nabla(\varepsilon\theta_{k}^{kl})\\
\nabla P_{n}^{kl}&=&0+\nabla_y N_n^{kl}(y)+\nabla(\varepsilon\theta_{n}^{kl})=\mathbf{W_{n}^{kl}}+\nabla(\varepsilon\theta_{n}^{kl})\\
\nabla \Phi_{kl}&=&0+\nabla_y \phi_{kl}(y)+\nabla(\varepsilon\Theta_{kl})=\mathbf{M_{kl}}+\nabla(\varepsilon\Theta_{kl})
 \end{eqnarray}
where $(\mathbf{\theta}^{kl},\Theta_{kl})$ satisfies the following equation system
\begin{align}\label{hmm proof2(8)}
      \left\{\begin{array} {l@{\quad}l} \displaystyle
          -\frac{\partial}{\partial{x_{j}}}(c_{ijnh}^{(\varepsilon)}\frac{\partial{\theta_{n}^{kl}}}{\partial{x_{h}}}+e_{hij}^{(\varepsilon)}\frac{\partial{\Theta_{kl}}}{\partial{x_{h}}})=0\qquad
          &\mbox{
            $x\in I_{\delta}(x_{\alpha})$}\\\\
          \displaystyle
          \frac{\partial}{\partial{x_{i}}}(e_{inh}^{(\varepsilon)}\frac{\partial{\theta_{n}^{kl}}}{\partial
            x_{h}}-\epsilon_{ih}^{(\varepsilon)}\frac{\partial{\Theta_{kl}}}{\partial{x_{h}}})=0\qquad
          &\mbox{
            $x\in I_{\delta}(x_{\alpha})$}\\\\
          \displaystyle \theta_{n}^{kl}=-N_{n}^{kl}\qquad&\mbox{ $x\in\partial
            I_{\delta}(x_{\alpha})$}\\\\
          \displaystyle \Theta_{kl}=-\phi_{kl}\qquad&\mbox{ $x\in\partial
            I_{\delta}(x_{\alpha})$}
        \end{array}\right.
    \end{align}
(\ref{hmm ce2})-(\ref{hmm ela2}), we have
\begin{align}
\displaystyle&c_{ijkl}^{H}-c_{ijkl}^{(0)}\nonumber\\
\displaystyle=&\langle c^{(\varepsilon)}_{smnh}((W_{n}^{kl})_h+\varepsilon\frac{\partial\theta_{n}^{kl}}{\partial x_h}))((W_{s}^{ij})_m+\varepsilon\frac{\partial\theta_{s}^{ij}}{\partial x_m}))\nonumber\\
\displaystyle&+e_{smh}((M_{kl})_h+\varepsilon\frac{\partial \Theta_{kl}}{\partial x_h})((W_{s}^{ij})_m+\varepsilon\frac{\partial\theta_{s}^{ij}}{\partial x_m}))\rangle_{I_\delta}\nonumber\\
\displaystyle&-\langle c_{smnh}(y)(W_{n}^{kl})_h(W_{s}^{ij})_m+e_{smh}(y)(M_{kl})_h(W_{s}^{ij})_m\rangle_{Y}\label{hmm cc11}\\
\displaystyle=&\langle c^{(\varepsilon)}_{smnh}(y)(W_{n}^{kl})_h(W_{s}^{ij})_m+e^{(\varepsilon)}_{smh}(y)(M_{kl})_h(W_{s}^{ij})_m\rangle_{I_\delta}\nonumber\\
\displaystyle&-\langle c_{smnh}(y)(W_{n}^{kl})_h(W_{s}^{ij})_m+e_{smh}(y)(M_{kl})_h(W_{s}^{ij})_m\rangle_{Y}\nonumber\\
\displaystyle&+\varepsilon\langle(c^{(\varepsilon)}_{smnh}(W_{n}^{kl})_h+e^{(\varepsilon)}_{smh}(M_{kl})_h)\frac{\partial\theta_{s}^{ij}}{\partial x_m}\rangle_{I_\delta}\nonumber\\
\displaystyle&+\varepsilon\langle c^{(\varepsilon)}_{smnh}\frac{\partial\theta_{n}^{kl}}{\partial x_h}(W_{s}^{ij})_m+e^{(\varepsilon)}_{smh}\frac{\partial \Theta_{kl}}{\partial x_h}(W_{s}^{ij})_m\rangle_{I_\delta}\nonumber\\
\displaystyle&+\varepsilon^{2}\langle c^{(\varepsilon)}_{smnh}\frac{\partial\theta_{n}^{kl}}{\partial x_h}\frac{\partial\theta_{s}^{ij}}{\partial x_m}+e^{(\varepsilon)}_{smh}\frac{\partial \Theta_{kl}}{\partial x_h}\frac{\partial\theta_{s}^{ij}}{\partial x_m}\rangle_{I_\delta}\label{hmm cc12}
 \end{align}
(\ref{hmm ce3})-(\ref{hmm ela3}), we obtain
\begin{align}
&c_{ijkl}^{H}-c_{ijkl}^{(0)}\nonumber\\
=&\langle c^{(\varepsilon)}_{smnh}((W_{n}^{kl})_h+\varepsilon\frac{\partial\theta_{n}^{kl}}{\partial x_h}))((W_{s}^{ij})_m+\varepsilon\frac{\partial\theta_{s}^{ij}}{\partial x_m}))\nonumber\\
&+\epsilon^{(\varepsilon)}_{hm}((M_{ij})_m+\varepsilon\frac{\partial \Theta_{ij}}{\partial x_m})((M_{kl})_h+\varepsilon\frac{\partial \Theta_{kl}}{\partial x_h})\rangle_{I_\delta}\nonumber\\
&-\langle c_{smnh}(y)(W_{n}^{kl})_h(W_{s}^{ij})_m+\epsilon_{hm}(y)(M_{ij})_m(M_{kl})_h\rangle_{Y}\label{hmm cc21}\\
=&\langle c^{(\varepsilon)}_{smnh}(W_{n}^{kl})_h(W_{s}^{ij})_m+\epsilon^{(\varepsilon)}_{hm}(M_{ij})_m(M_{kl})_h\rangle_{I_\delta}\nonumber\\
&-\langle c_{smnh}(y)(W_{n}^{kl})_h(W_{s}^{ij})_m+\epsilon_{hm}(y)(M_{ij})_m(M_{kl})_h\rangle_{Y}\nonumber\\
&+2\varepsilon\langle c^{(\varepsilon)}_{smnh}\frac{\partial\theta_{n}^{kl}}{\partial x_h}(W_{s}^{ij})_m+\epsilon^{(\varepsilon)}_{hm}(M_{ij})_m\frac{\partial\Theta_{kl}}{\partial x_h}\rangle_{I_\delta}\nonumber\\
&+\varepsilon^{2}\langle c^{(\varepsilon)}_{smnh}\frac{\partial\theta_{n}^{kl}}{\partial x_h}\frac{\partial\theta_{s}^{ij}}{\partial x_m}+\epsilon^{(\varepsilon)}_{hm}\frac{\partial\Theta_{ij}}{\partial x_m}\frac{\partial\Theta_{kl}}{\partial x_h}\rangle_{I_\delta}\label{hmm cc22}
 \end{align}
$2\times(\ref{hmm cc12})-(\ref{hmm cc22})$, we have
\begin{align}
&c_{ijkl}^{H}-c_{ijkl}^{(0)}\nonumber\\
=&\langle c^{(\varepsilon)}_{smnh}(y)(W_{n}^{kl})_h(W_{s}^{ij})_m+e^{(\varepsilon)}_{smh}(y)(M_{kl})_h(W_{s}^{ij})_m\rangle_{I_\delta}\nonumber\\
-&\langle c_{smnh}(y)(W_{n}^{kl})_h(W_{s}^{ij})_m+e_{smh}(y)(M_{kl})_h(W_{s}^{ij})_m\rangle_{Y}\nonumber\\
+&\varepsilon^{2}(2\langle c^{(\varepsilon)}_{smnh}\frac{\partial\theta_{n}^{kl}}{\partial x_h}\frac{\partial\theta_{s}^{ij}}{\partial x_m}+e^{(\varepsilon)}_{smh}\frac{\partial \Theta_{kl}}{\partial x_h}\frac{\partial\theta_{s}^{ij}}{\partial x_m}\rangle_{I_\delta}-\langle c^{(\varepsilon)}_{smnh}\frac{\partial\theta_{n}^{kl}}{\partial x_h}\frac{\partial\theta_{s}^{ij}}{\partial x_m}+\epsilon^{(\varepsilon)}_{hm}\frac{\partial\Theta_{ij}}{\partial x_m}\frac{\partial\Theta_{kl}}{\partial x_h}\rangle_{I_\delta})\nonumber\\
+&2\varepsilon (\langle (e^{(\varepsilon)}_{smh}(W_{s}^{ij})_m-\epsilon^{(\varepsilon)}_{mh}(M_{ij})_m)\frac{\partial\Theta_{kl}}{\partial x_h}\rangle_{I_\delta}+\langle(c^{(\varepsilon)}_{smnh}(W_{n}^{kl})_h+e^{(\varepsilon)}_{smh}(M_{kl})_h)\frac{\partial\theta_{s}^{ij}}{\partial x_m}\rangle_{I_\delta})
 \end{align}
Setting
\begin{align}
I_0(\varepsilon^0)
=&\langle c^{(\varepsilon)}_{smnh}(y)(W_{n}^{kl})_h(W_{s}^{ij})_m+e^{(\varepsilon)}_{smh}(y)(M_{kl})_h(W_{s}^{ij})_m\rangle_{I_\delta}\nonumber\\
-&\langle c_{smnh}(y)(W_{n}^{kl})_h(W_{s}^{ij})_m+e_{smh}(y)(M_{kl})_h(W_{s}^{ij})_m\rangle_{Y}\nonumber\\
I_1(\varepsilon^1)=&2\varepsilon \langle (e^{(\varepsilon)}_{smh}(W_{s}^{ij})_m-\epsilon^{(\varepsilon)}_{mh}(M_{ij})_m)\frac{\partial\Theta_{kl}}{\partial x_h}\rangle_{I_\delta}+2\varepsilon \langle(c^{(\varepsilon)}_{smnh}(W_{n}^{kl})_h+e^{(\varepsilon)}_{smh}(M_{kl})_h)\frac{\partial\theta_{s}^{ij}}{\partial x_m}\rangle_{I_\delta}\nonumber\\
I_2(\varepsilon^2)=&2\varepsilon^{2}\langle c^{(\varepsilon)}_{smnh}\frac{\partial\theta_{n}^{kl}}{\partial x_h}\frac{\partial\theta_{s}^{ij}}{\partial x_m}+e^{(\varepsilon)}_{smh}\frac{\partial \Theta_{kl}}{\partial x_h}\frac{\partial\theta_{s}^{ij}}{\partial x_m}\rangle_{I_\delta}-\varepsilon^{2}\langle c^{(\varepsilon)}_{smnh}\frac{\partial\theta_{n}^{kl}}{\partial x_h}\frac{\partial\theta_{s}^{ij}}{\partial x_m}+\epsilon^{(\varepsilon)}_{hm}\frac{\partial\Theta_{ij}}{\partial x_m}\frac{\partial\Theta_{kl}}{\partial x_h}\rangle_{I_\delta}
 \end{align}
For $I_0(\varepsilon^0)$, referring to the ideas of Lemma \ref{hom lemone}, Lemma \ref{hom lemtwo} and \citet{x.h.wu}, we get
\begin{eqnarray}
\mid I_0(\varepsilon^0)\mid\leq C_0\frac{\varepsilon}{\delta}
\end{eqnarray}
For $I_1(\varepsilon^1)$, we denote the first term and the second term of $I_1(\varepsilon^1)$ as $I_1(1)$ and $I_1(2)$ respectively.\\
By Green Formula, we have
\begin{eqnarray*}
I_1(1)=\displaystyle&\displaystyle\frac{2\varepsilon}{\mid I_\delta\mid}(\displaystyle\int_{\partial I_\delta}\Theta_{kl}(e^{(\varepsilon)}_{smh}(W_{s}^{ij})_m-\epsilon^{(\varepsilon)}_{mh}(M_{ij})_m)\cdot\overrightarrow{n}ds\nonumber\\
\displaystyle&-\displaystyle\int_{I_\delta}\Theta_{kl}\frac{\partial}{\partial x_h}(e^{(\varepsilon)}_{smh}(W_{s}^{ij})_m-\epsilon^{(\varepsilon)}_{mh}(M_{ij})_m)dx
 \end{eqnarray*}\\
Due to (\ref{asy cella}) and  (\ref{hmm proof2(8)}), we obtain,
\begin{eqnarray*}
\mid I_1(1)\mid=\mid\frac{2\varepsilon}{\mid I_\delta\mid}\int_{\partial I_\delta}\phi_{kl}(e^{(\varepsilon)}_{smh}(W_{s}^{ij})_m-\epsilon^{(\varepsilon)}_{mh}(M_{ij})_m)\cdot\overrightarrow{n}ds\mid
\end{eqnarray*}
Using the singularity of the solutions of (\ref{asy cella}) and $H\ddot{o}lder$ inequality, we get
\begin{eqnarray*}
\mid I_1(1)\mid
\leq \frac{2\varepsilon}{\mid I_\delta\mid}C_1\mid\partial I_\delta\mid^{\frac{1}{2}}\leq C_1 \frac{\varepsilon}{\delta}
 \end{eqnarray*}
Similarly, we can get
\begin{eqnarray*}
\mid I_1(2)\mid\leq C_2 \frac{\varepsilon}{\delta}
\end{eqnarray*}
Thus,
\begin{eqnarray}
\mid I_1(\varepsilon^1)\mid\leq max(C_1,C_2)\frac{\varepsilon}{\delta}
\end{eqnarray}
For $I_2(\varepsilon^2)$, we may just follow the steps in the proof of Lemma \ref{hom lemthree} to get the following estimate,
\begin{eqnarray}
\mid I_2(\varepsilon^2)\mid\leq C_3\frac{\varepsilon}{\delta}
\end{eqnarray}
Therefore, we obtain
\begin{eqnarray*}
e_{c}(HMM)\leq max\{C_0,C_1,C_2,C_3\}\frac{\varepsilon}{\delta}= C_c\frac{\varepsilon}{\delta}
\end{eqnarray*}
Similarly, we can also get
\begin{eqnarray*}
e_{d}(HMM)\leq C_{d}\frac{\varepsilon}{\delta}
\end{eqnarray*}
For the estimation of $e_{e}(HMM)$, we do not have essential difference with the estimation of $e_{c}(HMM)$ and $e_{d}(HMM)$, except some techniques on constructing suitable expressions for analytical convenience, which is not necessary to be writen out in details here.\\
This completes the proof of Theorem \ref{hmmthm2}
\end{proof}

\section{Numerical simulation}
\subsection{Numerical example}
In order to illustrate the multi-scale methods mentioned above in this paper, we introduce an numerical example of 1-3 type piezoelectric composite materials made of piezoceramic(PZT) fibers embedded in a soft non-piezoelectric matrix(polymer). Since, for a transversely isotropic piezoelectric solid, the stiffness matrix, the piezoelectric matrix and the dielectric matrix simplify so that there remain 11 independent coefficients, we take the piezoelectric composite materials aligned fibers made of a transversely isotropic piezoelectric solid (PZT), embedded in an isotropic polymer matrix. Moreover, it is easy to verify that the resulting composite is a transversely isotropic piezoelectric material too. Therefore, if we take
 \begin{align*}
11\longrightarrow 1\qquad 22\longrightarrow 2\qquad 33\longrightarrow3\qquad23\longrightarrow4 \qquad 13\longrightarrow 5\qquad 12\longrightarrow 6
\end{align*}
then the constitutive relations of both each component and the resulting composite have the following form.
 \begin{equation*}
        \left(\begin{array}{c}
            \sigma_{11}\\
            \sigma_{22}\\
            \sigma_{33}\\
            \sigma_{23}\\
            \sigma_{31}\\
            \sigma_{12}\\
            D_{1}\\
            D_{2}\\
            D_{3}
          \end{array}\right)=\left(\begin{array}{ccccccccc}
            c_{11}& c_{12} &c_{13} &0 &0 &0 &0 &0 &-e_{13}\\
            c_{12}& c_{22} &c_{13} &0 &0 &0 &0 &0 &-e_{13}\\
            c_{13}& c_{13} &c_{33} &0 &0 &0 &0 &0 &-e_{33}\\
            0& 0& 0& c_{44}&0 &0 &0 &-e_{15} &0\\
            0& 0& 0& 0&c_{44}&0&-e_{15}&0&0 \\
            0& 0& 0& 0&0&c_{66}&0&0&0\\
            0& 0& 0& 0&e_{15}&0&\epsilon_{11}&0&0 \\
            0& 0& 0& e_{15}&0 &0 &0 &\epsilon_{11} &0\\
            e_{13}& e_{13} &e_{33} &0 &0 &0 &0 &0 &\epsilon_{33}
          \end{array}\right)
        \left(\begin{array}{c}
            s_{11}\\
            s_{22}\\
            s_{33}\\
            s_{23}\\
            s_{31}\\
            s_{12}\\
            E_{1}\\
            E_{2}\\
            E_{3}
          \end{array}\right)
      \end{equation*}

In our numerical example, we take the material properties of the composite constituents fiber (PZT-5) and matrix (polymer) as Table 1.
\begin{table}[htb]
\caption{\label{table.label}\footnotesize{Material properties of the composite constituents fiber(PZT-5) and matrix(polymer)}}
\centering
\bigskip
\begin{scriptsize}
\begin{tabular}{rlllllllllll}
\hline\\[-1.0 ex]
 & $c_{11}(10^{10})$ & $c_{12}(10^{10})$ & $c_{13}(10^{10})$ & $c_{33}(10^{10})$ & $c_{44}(10^{10})$ & $c_{66}(10^{10})$&$e_{15}$&$e_{13}$&$e_{33}$&$\epsilon_{11}(10^{-9})$&$\epsilon_{33}(10^{-9})$\\[0.75 ex]\hline\\[-0.75 ex]
PZT-5 & 12.1 & 7.54 & 7.52 & 11.1 & 2.11 & 2.28&12.3&-5.4&15.8&8.11&7.35\\[0.75 ex]
Polymer & 0.386 & 0.257 & 0.257 & 0.386 & 0.064 & 0.064&--&--&--&0.07965&0.07965
\\\hline
\end{tabular}
\end{scriptsize}
\end{table}

And the multi-scale structure of the composite is shown in .

We take $\Omega=[0,5]\times[0,5]\times[0,5]$ as the computation domain, and $R$ denotes the radius of the fiber in one microscopic periodic. $\varepsilon=0.015625$ is the microscopic periodic. Ignoring volume forces, we apply boundary condition as follows,
\begin{align*}
u_{i}=0\qquad\varphi(5,y,z)=1000\qquad\varphi(x,y,z)=0\quad for\quad x\neq 5,\qquad\mbox{on $ \partial\Omega$}
\end{align*}

The computation domain $\Omega$ is divided into $40\times 40\times 40$ elements, the coarse mesh. Each element contains  $8\times 8\times 8$ microscopic periodicals, see Picture 2. The computation process for given R and  $\delta$ is as follows:
\begin{enumerate}
\item Solving the two cell problems (\ref{asy cella}) and (\ref{asy cellb}) and calculating the effective coefficients as (\ref{asy coefficient}).
\item Solving (\ref{macro equperiodic}) on $\Omega$ with the effective coefficients got in Step 1.
\item Given $\alpha$, solving the cell problems of HMM (\ref{hmm modone}) and (\ref{hmm modtwo}) on $I_\delta(x_\alpha)$. Then we give out the estimation of the effective coefficients at scale H by (\ref{hmm coef1}) and (\ref{hmm coef2}). In this numerical example, we take one sample in each Macroscopic element K.
\item Solving (\ref{hmm macroweak}) on $\Omega$ by the estimation of the effective coefficients got in Step 3.
\end{enumerate}
\setlength{\unitlength}{1.0cm}
\begin{picture}(6,6)
\put(-1.5,-1)
{
\setlength{\unitlength}{1.0cm}
\begin{picture}(10,10)
\multiput(0,0)(0,0.125){40}
{
\setlength{\unitlength}{1.0cm}
\begin{picture}(5,0.5)
\multiput(0,0)(0.125,0){40}
{
\framebox(0.125,0.125){}
}
\end{picture}
}
\put(0.4,5){\line(1,1){2.0}}
\put(5.4,5){\line(1,1){2.0}}
\put(5.4,0){\line(1,1){2.0}}
\put(2.4,7.0){\line(1,0){5}}
\put(7.4,7.0){\line(0,-1){5}}
\multiput(0.45,5.05)(0.05,0.05){39}{\line(1,0){5}}
\multiput(0.525,5.0)(0.125,0){39}{\line(1,1){2.0}}
\multiput(5.45,5.05)(0.05,0.05){39}{\line(0,-1){5}}
\multiput(5.4,0.125)(0,0.125){39}{\line(1,1){2.0}}
\end{picture}
}
\put(8,0.5)
{
\setlength{\unitlength}{0.5cm}
\begin{picture}(10,10)
\multiput(0,0)(0,0.625){8}
{
\setlength{\unitlength}{0.5cm}
\begin{picture}(5,0.625)
\multiput(-0.4,0)(0.625,0){8}
{
{\framebox(0.625,0.625){}}
\put(-0.3125,0.3125){\color{BrickRed}{\circle*{0.3125}}}
}
\end{picture}
}

\multiput(0.4,5.0)(0.25,0.25){9}{\line(1,0){5}}
\multiput(5.4,5.0)(-0.625,0){9}{\line(1,1){2.0}}
\multiput(7.4,7.0)(-0.25,-0.25){8}{\line(0,-1){5}}
\multiput(5.4,0)(0,0.625){9}{\line(1,1){2.0}}
\end{picture}
}

\put(13,1.0)
{
\setlength{\unitlength}{0.25cm}
\begin{picture}(10,10)

\put(0.4,0){\framebox(5.0,5.0){}}
\put(2.9,2.5){\color{BrickRed}{\circle*{50.0}}}

\put(0.4,5){\line(1,1){2.0}}
\put(5.4,5){\line(1,1){2.0}}
\put(5.4,0){\line(1,1){2.0}}
\put(2.4,7.0){\line(1,0){5}}
\put(7.4,7.0){\line(0,-1){5}}


\end{picture}
}
\put(3.0,1.05){\vector(1,0){5.3}}
\put(7,1.15){unit}
\put(6.8,0.65){element}
\put(10.55,1.2){\vector(1,0){2.7}}
\put(12.35,1.3){unit}
\put(12.40,0.85){cell}
\put(14,1.7){\line(0,1){1.5}}
\put(13.7,3.25){fiber}
\put(13.5,0.5){\line(1,1){0.7}}
\put(13,0.25){matrix}
\put(0,-3.5){Picture 2. The multi-scale structure of 1-3 type piezoelectric composite materials}
\end{picture}\\[4.4 ex]

\newpage
\subsection{Numerical results and analysis}
\subsubsection{Order of the effective coefficients}
For given $R=0.3125\varepsilon$ and $R=0.4375\varepsilon$, we take $\delta=4\varepsilon$, $\delta=3\varepsilon$, $\delta=2\varepsilon$ respectively. We take the scale of elements on $I_\delta$ as $h=0.00078125$. Then, following the steps in Section 3.1, we give out the order of error estimate between the effective coefficients got in Step 1 and the estimation of the effective coefficients got in Step 3.

\begin{table}[htbp]
\centering
\bigskip
\begin{scriptsize}
\begin{tabular}{cllllll}
  \hline\\[-2.4 ex]
$\frac{\varepsilon}{\delta}$ &$c^{H}_{11}$ & order& $c^H_{12}$ & order&$c^{H}_{13}$& order \\[0.75 ex] \hline
$1/4$& 1.11e+10   & 0.93     &6.18e+09   & 0.99 &6.70e+09  & 0.95\tabularnewline
$1/3$& 1.23e+10   &0.97 &6.99e+09   &1.03  &7.53e+09  &0.99   \tabularnewline
$1/2$& 1.51e+10   &--        & 8.70e+09  &--& 9.23e+09  &-- \tabularnewline
$c^{(0)}_{ij}$&0.67e+10  &--& 3.71e+09   &--& 4.10e+09  &--
\tabularnewline \hline
\end{tabular}
\end{scriptsize}
\bigskip
\begin{scriptsize}
\begin{tabular}{cllllll}
  \hline\\[-2.4 ex]
$\frac{\varepsilon}{\delta}$ &$c^H_{33}$ & order&$c^{H}_{44}$ & order& $c^H_{66}$ & order\\[0.75 ex]\hline
$1/4$& 2.28e+10   & 1.06& 2.25e+09          & 1.15& 2.08e+09      &0.92\\
$1/3$& 2.35e+10      &1.07 &2.61e+09      &1.26 & 2.38e+09   &0.94\\
$1/2$& 2.48e+10      &--&3.47e+09   &-- & 2.98e+09   &-- \\
$c^{(0)}_{ij}$  & 2.11e+10  &--& 1.33e+09  &--&1.10e+09      &--
\\\hline
\end{tabular}
\end{scriptsize}

\bigskip
\begin{scriptsize}
\begin{tabular}{cllllll}
  \hline\\[-2.4 ex]
$\frac{\varepsilon}{\delta}$ &$e^{H}_{13}$ & order& $e^H_{33}$ & order&$e^{H}_{15}$ & order\\[0.75 ex]\hline
$1/4$& - 0.31       & 0.95&   5.88      & 1.06& 0.46      &1.18 \\
$1/3$&- 0.37      &0.99&     5.83      &1.07& 0.63  &1.28\\
$1/2$&- 0.49  &--    &    5.74         &-- &  1.04       &--\\
$e^{(0)}_{ij}$&- 0.11   &--&    6.01       &--&0.04&--
\\\hline
\end{tabular}
\end{scriptsize}
\bigskip
\begin{scriptsize}
\begin{tabular}{cllllll}
  \hline\\[-2.4 ex]
$\frac{\varepsilon}{\delta}$ &$\epsilon^{H}_{11}$ & order& $\epsilon^{H}_{33}$ & order \\[0.75 ex]\hline
$1/4$&6.40e-10  &0.83&2.38e-09  &  1.03\\
$1/3$& 7.67e-10&0.78 &  2.38e-09 &1.05\\
$1/2$& 9.90e-10&-- &  2.38e-09  &-- \\
$\epsilon^{(0)}_{ij}$&1.72e-10&--&  2.40e-09&--
\\\hline
\end{tabular}
\end{scriptsize}
\caption{\label{table1} \footnotesize{Accuracy of HMM on coefficients for 1-3 type composites with $R=0.3125\varepsilon$}}
\end{table}

\begin{table}[htbp]
\centering
\bigskip
\begin{scriptsize}
\begin{tabular}{cllllll}
  \hline\\[-2.4 ex]
$\frac{\varepsilon}{\delta}$ &$c^{H}_{11}$ & order& $c^H_{12}$ & order&$c^{H}_{13}$ & order\\[0.75 ex]\hline
$1/4$&2.31e+10   & 1.01 &1.00e+10      & 1.07  &1.28e+10  &  1.04\\
$1/3$&2.55e+10   &1.03  &1.17e+10      &1.09 & 1.43e+10  &1.05\\
$1/2$&3.05e+10   &--    & 1.51e+10     &--& 1.76e+10   &--\\
$c^{(0)}_{ij}$&1.58e+10   &--& 0.55e+10   &--& 0.83e+10   &--
\\\hline
\end{tabular}
\end{scriptsize}
\bigskip
\begin{scriptsize}
\begin{tabular}{cllllll}
  \hline\\[-2.4 ex]
$\frac{\varepsilon}{\delta}$&$c^H_{33}$ & order &$c^{H}_{44}$ & order& $c^H_{66}$ & order\\[0.75 ex]\hline
$1/4$&4.27e+10  &     1.17& 5.38e+09   & 1.31    &4.08e+09        &0.98\\
$1/3$&4.39e+10  &1.15 &  6.22e+09      &1.28 & 4.71e+09          &0.98  \\
$1/2$&4.64e+10  &--&8.03e+09   &--     & 5.94e+09         &--  \\
$c^{(0)}_{ij}$& 3.97e+10  &--& 3.56e+09    &--& 2.17e+09          &--
\\\hline
\end{tabular}
\end{scriptsize}
\end{table}
\bigskip
\begin{table}
\begin{scriptsize}
\begin{tabular}{cllllll}
  \hline\\[-2.4 ex]
$\frac{\varepsilon}{\delta}$ &$e^{H}_{13}$ & order& $e^{H}_{33}$ & order&$e^{H}_{15}$ & order\\[0.75 ex]\hline
$1/4$&-0.76  &1.03&11.40    &   1.17&  1.07       & 1.23\\
$1/3$& -0.88     &1.05  &   11.31      &1.15& 1.49 &1.24 \\
$1/2$&  -1.12   &--&  11.13     &--&  2.42       &--   \\
$e^{(0)}_{ij}$& -0.42 &--&   11.63     &--&    &--
\\\hline
\end{tabular}
\end{scriptsize}
\centering
\bigskip
\begin{scriptsize}
\begin{tabular}{cllllll}
  \hline\\[-2.4 ex]
$\frac{\varepsilon}{\delta}$ &$\epsilon^{H}_{11}$ & order& $\epsilon^{H}_{33}$ & order\\[0.75 ex]\hline
$1/4$&1.35e-09  & 0.84&4.59e-09  &   1.11\\
$1/3$&  1.60e-09  &0.80 &  4.59e-09  &1.10\\
$1/2$& 2.05e-09&--  & 4.57e-09  &--  \\
$\epsilon^{(0)}_{ij}$& 4.32e-10&--&   4.61e-09&--
\\\hline
\end{tabular}
\end{scriptsize}
\caption{\label{table2} \footnotesize{Accuracy of HMM on coefficients for 1-3 type composites with $R=0.4375\varepsilon$}}
\end{table}
\newpage
\subsubsection{Analysis of the results}

In the numerical result, we can see that the order of error of the effective coefficients we have got in the numerical examples support the corresponding conclusions in our theoretical analysis in section 5. Therefore, we can predict that the estimation of the effective coefficients in the multi-scale model we designed under the framework of HMM can be better and better as the number of the cells taken in one sample increases larger and larger.
\section{Conclusions}
In this paper, a multi-scale model for piezoelectric composite materials under the framework of Heterogeneous Multi-scale Method(HMM) is proposed. In order to verify the capability of the multi-scale model we developed, macroscopic model is derived from microscopic
model of piezoelectric composite material by asymptotic expansion for materials with periodic microstructure. Convergence rate of asymptotic expansion is proved to be $\sqrt{\epsilon}$ under the framework of homogenization theory. We then give out both error
estimate between HMM solutions and homogenization solutions and error estimate of the effective coefficients for piezoelectric composite materials with periodic microstructure. Therefore, HMM solutions is shown to convergence to homogenization solutions in the order of $(H+\frac{\epsilon}{\delta})$ and the effective coefficients got by HMM modeling is shown to convergence to effective coefficients got by homogenization theory in the order of $\frac{\epsilon}{\delta}$. Moreover, our numerical simulation result support the corresponding theoretical conclusions we got above very well.
\section*{Acknowledgement}
We would like to thank Pingbing Ming and Felix Otto for helpful discussions.

\end{document}